\documentclass[11pt, reqno]{extarticle}
\usepackage[utf8]{inputenc}
\usepackage[margin=1in]{geometry}
\usepackage{enumitem, amssymb, amsmath, amsthm,xcolor, cancel, hyperref, graphicx, etoolbox, setspace, subcaption, lmodern, yfonts, float, subcaption, colortbl, dsfont, centernot, comment, bigints, empheq,amsthm}
\usepackage{authblk}
\usepackage[font=small,labelfont=bf,width=0.95\textwidth]{caption}
\usepackage{rotating}
\usepackage{indentfirst}
\usepackage{appendix}
\newtheorem{theorem}
{Theorem}
\newtheorem{corol}{Corollary}
\newtheorem{remark}{Remark}

\def\Abar{\bar{A}}

\def\rbar{\bar{\rho}}
\def\ubar{\bar{u}}

\allowdisplaybreaks

\counterwithin{equation}{subsection}
\counterwithin{theorem}{section}
\counterwithin{remark}{section}
\counterwithin{corol}{section}
\title{

Global Bifurcations and Pattern Formation in Target–Offender–Guardian Crime Models

}
\author[1]{Madi Yerlanov}
\author[2]{Qi Wang}
\author[1]{Nancy Rodr{\'\i}guez}
\affil[1]{University of Colorado Boulder, Department of Applied Mathematics,\authorcr 1111 Engineering Center, Boulder, CO, 80309}
\affil[2]{Howard University, Department of Mathematics,\authorcr 204 Academic Service Building B,
Washington, D.C., 20059}
\date{\today}
\begin{document}
\maketitle

\begin{abstract}
We study a reaction–advection–diffusion model of a target–offender–guardian system designed to capture interactions between urban crime and policing. Using Crandall--Rabinowitz bifurcation theory and spectral analysis, we establish rigorous conditions for both steady-state and Hopf bifurcations. These results identify critical thresholds of policing intensity at which spatially uniform equilibria lose stability, leading either to persistent heterogeneous hotspots or oscillatory crime–policing cycles. From a criminological perspective, such thresholds represent tipping points in guardian mobility: once crossed, they can lock neighborhoods into stable clusters of criminal activity or trigger recurrent waves of hotspot formation. Numerical simulations complement the theory, exhibiting stationary patterns, periodic oscillations, and chaotic dynamics. By explicitly incorporating law enforcement as a third interacting component, our framework extends classical two-equation models. It offers new tools for analyzing nonlinear interactions, bifurcations, and pattern formation in multi-agent social systems.
\end{abstract}

\section{Introduction}

We study an initial-boundary value problem for $(A, \rho, u)$, defined on the space-time domain $(\textbf{x}, t) \in \Omega \times (0, T]$, where $\Omega \subset \mathbb{R}^N$ is a bounded domain with piecewise smooth boundary $\partial \Omega$, and the time horizon $T>0$
\begin{equation}\label{eq:E}
    \tag{E}
    \left\{
    \begin{aligned}
        \frac{\partial A}{\partial t}=&D_A\Delta A - A+\alpha+A\rho, &&\text{ in } \Omega\times(0, T], \\
        \frac{\partial \rho}{\partial t}=&\nabla \cdot \left(D_\rho \nabla \rho-2 \rho\nabla \ln A\right)-A\rho+\beta - \rho u, &&\text{ in } \Omega\times(0, T], \\
        \frac{\partial u}{\partial t}=&\nabla \cdot \left(D_u \nabla u-\chi u\nabla \ln A\right), &&\text{ in } \Omega\times(0, T], \\
         \frac{\partial A}{\partial \textbf{n}}=& \frac{\partial \rho}{\partial \textbf{n}}= \frac{\partial u}{\partial \textbf{n}}=0, &&\text{ on }\partial\Omega, \\
         (A, \rho, u)&(\textbf{x}, 0)=(A_0, \rho_0, u_0)\geq, \not\equiv 0, && \text{ in }\Omega.
    \end{aligned}
    \right.
\end{equation}
Here, $\partial$ denotes partial differentiation with respect to time $t$, while $\nabla$ and $\Delta$ represent the spatial gradient and Laplacian operators, respectively. The parameters $D_A, D_\rho, D_u, \alpha, $ and $\beta$ are assumed to be positive constants, whereas $\chi$ may take either positive or negative values, depending on the modeling context, as discussed later.  The zero-flux (Neumann) boundary conditions, expressed via $\frac{\partial}{\partial \mathbf{n}}$, are imposed to model an enclosed environment. The initial conditions are assumed to be nonnegative and not identically zero; this, together with the strong maximum principle, ensures that each component $(A, \rho, u)$ remains strictly positive for all $t > 0$. Finally, the terminal time $T$ represents the time horizon for the analysis: it is taken as infinite in theoretical investigations, and sufficiently large in numerical simulations to capture the system’s long-term behavior.

\subsection{The UCLA model}

System \eqref{eq:E} can be used to model the spatio-temporal evolution of urban criminal activity under the influence of police patrol intervention.  Specifically, it describes the dynamics of the population density $\rho$ of criminal agents, who are attracted to regions with high attractiveness $A$ and are deterred by the presence of law enforcement, represented by $u$.  This approach to modeling urban crime traces back to the seminal work \cite{short2008statistical} of a UCLA team in 2008, which was inspired by theories such as routine activity theory, the broken-windows effect, and the repeat and near-repeat victimization effect \cite{sampson2004seeing, windows1982police}. They developed a two-dimensional agent-based lattice model that captures both the movement of offenders and the evolution of local attractiveness.  In the spatial continuum limit, this discrete model leads to a fully coupled system of partial differential equations:
\begin{equation} \label{2008model}
    \left\{
    \begin{aligned}
        \frac{\partial A}{\partial t}=&D_A\Delta A - A+A_0+A\rho, &&\text{ in } \Omega\times(0, T], \\
        \frac{\partial \rho}{\partial t}=&\nabla \cdot \left(D_\rho \nabla \rho-2 \rho\nabla \ln A\right)-A\rho+\bar B - \rho u, &&\text{ in } \Omega\times(0, T].
    \end{aligned}
    \right.
\end{equation}
The linear diffusion coefficient $D_A$ represents the effect of the near-repeat victimization phenomenon. Since the expected number of crimes is given by the product $A\rho$, the self-exciting nature of crime is incorporated into the dynamics through the $+A\rho$ term in the first equation.  When a crime occurs, it is assumed that the perpetrator is removed, leading to the $-A\rho$ term in the second equation. External sources of growth, which may vary in space and time, are also considered for both variables and are denoted by $A_0$ and $\bar{B}$, respectively.  Finally, we note that the movement of offenders combines random motion, modeled by linear diffusion, with a directional bias toward regions of high attractiveness—resulting in a chemotaxis-like drift. We refer interested readers to \cite{gu2017stationary, short2008statistical} for detailed derivations and further justification of system \eqref{2008model}, and to \cite{d2015statistical, groff2019state} for reviews of agent-based models in urban crime modeling.

The UCLA continuum model and its 2D lattice counterpart have gained significant academic recognition for their ability to simulate complex spatiotemporal dynamics and generate both regular and irregular spatial patterns.  Most notably, these models successfully capture the aggregation phenomena, commonly known as hotspots, in urban criminal activity, providing critical insights into the spatial distribution of crime.  The global well-posedness of the model in \eqref{2008model} is investigated in \cite{freitag2018global, jiang2024global, rodriguez2013global, rodriguez2022global, wang2020global, winkler2019global}.  These studies have established that the system admits a unique global-in-time solution, which remains uniformly bounded in one-dimensional settings or in higher dimensions when the advection effect is weak or sublinear.  More importantly, a substantial body of research, including \cite{berestycki2014existence, cantrell2012global, garcia2013existence, gu2017stationary, heihoff2020generalized, kolokolnikov2012stability, lloyd2013localised, lloyd2016exploring, manasevich2013global, mei2020existence, short2010nonlinear,short2010dissipation,  tse2016hotspot}, has focused on the emergence of nontrivial spatial patterns. These works demonstrate that system \eqref{2008model} can produce structured crime hotspots, represented by spatially concentrated profiles that may be either stationary or time-dependent, thereby successfully capturing the characteristic features of urban residential burglary.  
Building upon these foundations, researchers have substantially extended the theoretical framework of system \eqref{2008model} to model real-world crime dynamics more accurately.  They include the study of nonlinear or anomalous diffusion and spatial heterogeneity in both the near-repeat victimization effect and the dispersal strategies of offenders \cite{gu2017stationary, martinez2022bi}; modeling criminal movement using L\'evy processes \cite{chaturapruek2013crime, levajkovic2016levy, pan2018crime}; incorporating age-structured populations \cite{kumar2023modelling, saldana2018age}; and examining geographic profiling \cite{mohler2012geographic}.  For further developments and comprehensive reviews in this area, we refer the reader to \cite{berestycki2010self, berestycki2014existence, d2015statistical, mohler2011self}.

\subsection{The UCLA model with Policing Force}

Given that crime and disorder are unevenly distributed across urban areas, scholars and practitioners in \cite{eck2003preventing, frydl2004fairness, sherman1995general, skogan1999community, weisburd2004can, weisburd1995policing} have advocated for concentrating policing efforts in high-demand zones rather than dispersing resources uniformly across the city.  This strategy, commonly referred to as “hot-spot policing” or “hotspotting, ” emphasizes geographic precision, focusing on the location of crimes rather than the characteristics of individual offenders.  While this targeted approach has demonstrated measurable short-term reductions in crime, its broader effectiveness remains the subject of debate.  Critics (\textit{cf.}, \cite{braga2001effects, braga2014effects, frydl2004fairness, kleck2014more, rosenbaum2006limits, weisburd2010problem}) have raised concerns about potential unintended consequences, such as the spatial displacement of crime, the erosion of community trust, and the risk of over-policing marginalized neighborhoods.  Moreover, questions remain about how different policing strategies interact with the underlying social and economic dynamics that drive criminal behavior. These challenges highlight the need for more comprehensive models that not only track where crime occurs but also capture how policing interventions influence the evolving spatial patterns of crime over time.

A major extension of the original model proposed by the UCLA team involves incorporating the effects of policing strategies on the evolution of criminal activity. In particular, the authors of \cite{jones2010statistical, pitcher2010adding} independently introduced an additional equation to represent police deployment and targeted patrols:
\begin{equation*}
u_t=\nabla \cdot (D_u\nabla u-\chi u\nabla\ln A), 
\end{equation*}
where $u$ denotes the distribution of law enforcement in the field.  The positive parameters $D_u$ and $\chi$ represent the intensity of the random movement of police and their response to criminal activity, respectively.  Incorporating an explicit law enforcement component allows one to systematically study how various patrol strategies, whether reactive or proactive, affect hotspot formation, stability, and suppression under different assumptions about agent movement and crime-attractiveness feedback.  When $\chi > 0$, law enforcement agents tend to migrate toward areas with high attractiveness, whereas $\chi < 0$ corresponds to a strategy in which police are deployed to regions of low attractiveness. The sign and magnitude of $\chi$ thus encode different policing strategies and can be chosen according to the specific modeling objectives.
 
While adding a third equation to model law enforcement enhances the system’s realism by incorporating police deployment strategies, it also introduces new mathematical challenges, particularly in analyzing pattern formation and stability. The resulting three-component reaction–diffusion system \eqref{eq:E} has been examined by several researchers to explore how policing influences crime dynamics. Tse and Ward \cite{tse2018asynchronous} conducted a detailed singular perturbation analysis of \eqref{eq:E} in one dimension, constructing steady-state hotspot (spike) solutions and performing a nonlocal eigenvalue analysis to characterize their linear stability. They distinguished between synchronous and asynchronous instabilities, showing that hotspot locations may destabilize either in phase (oscillating together) or out of phase (oscillating asynchronously), depending on police response and other system parameters. These results highlight how law enforcement strategies shape the temporal dynamics of urban crime distributions. Building on this foundation, Buttenschoen et al. \cite{buttenschoen2020cops} refined the linear stability framework and focused more explicitly on how targeted policing strategies influence hotspot selection and stabilization. Their work strengthens the connection between the number and spacing of spikes and the intensity of police deployment, demonstrating, for example, that increasing the strength or spatial focus of policing can shift the threshold for hotspot instability.

Rodr{\'\i}guez et al. \cite{rodriguez2021understanding} applied bifurcation theory to rigorously establish the existence of stable spatial and temporal patterns (steady states and periodic orbits), identify critical thresholds for pattern emergence, and analyze wavemode selection. Strikingly, they showed that anti-hotspot policing—dispatching police away from areas of high attractiveness—can stabilize crime hotspots and suppress oscillations. Although they proved global existence and boundedness of solutions, their analysis also revealed that the system may exhibit ill-posedness due to strong sensitivity to initial conditions, raising concerns about its predictive reliability. Further studies of system \eqref{eq:E} can be found in \cite{gai2024nucleation, ricketson2010continuum, yerlanov2025, zipkin2014cops}, where the extended model has been examined under different contexts. Collectively, these works demonstrate that the system supports a broad spectrum of dynamic behaviors—including spatially uniform equilibria, spatially heterogeneous steady states, periodic oscillations, and even chaotic patterns. Numerical simulations across the literature not only validate the theoretical predictions but also underscore the model’s inherent richness and complexity.

\subsection{Preliminaries and outline}

Before detailing our main steps and contributions, we first note that the framework assumes  a fixed number of guardians throughout all temporal dynamics as
\[\frac{d}{dt}\int_\Omega u(\textbf{x}, t)d\textbf{x}=0\quad \forall t>0.\]
If we introduce a parameter $\lambda$ to represent the spatial average of $A$, then the system admits the following positive constant steady state 
\begin{equation*}\label{eq:CSS}
    \tag{SS} 
    \Abar :=\lambda, \;\;
    \rbar :=1-\frac{\alpha}{\lambda}, \;\;
    \ubar :=\frac{1}{|\Omega|}\int_\Omega u_0(\textbf{x})\;d\textbf{x}= \lambda\left(\frac{\beta}{\lambda-\alpha}-1\right), 
\end{equation*}

each of which is positive under the condition 
\begin{equation*}
\alpha<\lambda<\alpha+\beta.
\end{equation*} 
Our study focuses on spatially heterogeneous solutions of system \eqref{eq:E} that extend beyond the homogeneous patterns described by \eqref{eq:CSS}.  To characterize the system's long-term behavior, we analyze its steady states, natural candidates for dynamical attractors, by seeking nonconstant positive solutions to the following coupled system:  
\begin{equation}\label{eq:EQ}\tag{EQ}
    \left\{
    \begin{aligned}
        &0=D_A\Delta A - A+\alpha+A\rho, &&\text{ in } \Omega, \\
             &0=\nabla \cdot \big(D_\rho \nabla \rho- 2 \rho \nabla \ln A\big)-A\rho+\beta - \rho u, &&\text{ in } \Omega, \\
        &0=\nabla \cdot \big(D_u \nabla u-\chi u\nabla \ln A\big), &&\text{ in } \Omega, \\
         &\frac{\partial A}{\partial \textbf{n}}= \frac{\partial \rho}{\partial \textbf{n}}= \frac{\partial u}{\partial \textbf{n}}=0, &&\text{ on }\partial\Omega.
    \end{aligned}
    \right.
\end{equation}

The interaction between spatial heterogeneity and nonlinear coupling creates significant mathematical challenges, and a deeper analysis is needed to fully understand the resulting dynamics. While prior studies in \cite{buttenschoen2020cops, gai2024nucleation, rodriguez2021understanding, tse2018asynchronous, zipkin2014cops} have shown the existence of spatially heterogeneous solutions, most focus on one-dimensional domains or simplified models, leaving open the problem of rigorous higher-dimensional analysis. This work addresses that gap by extending the framework of \cite{rodriguez2021understanding} to higher dimensions and by developing a general approach for bifurcation analysis in three-agent systems on spatial domains.
In particular, we establish the Crandall--Rabinowitz Theorem and stability results for systems with guardians in any spatial dimension, using Laplacian eigenvalues rather than wave-mode numbers.  Our proofs are given in explicit detail, making the methods adaptable to other three-agent systems. We also generalize the Hopf bifurcation Theorem to higher dimensions.
On the computational side, we provide parameter-based verification guided by linear stability analysis and include simulations in both one- and two-dimensional dimensions. These simulations highlight features such as long transients and consistent behaviors across dimensions. We further demonstrate chaotic dynamics using bifurcation diagrams and perform a systematic parameter sweep to explore how parameter magnitudes affect system outcomes. Together, these results provide a comprehensive and general treatment of the target-guardian-offender model.

\textit{Outline:} 

This work develops its analysis through three sequential components: first, a linear stability examination of the constant solution \eqref{eq:CSS} establishes the baseline conditions for pattern formation; second, a steady state bifurcation analysis reveals how intensive guardians drive the emergence of nontrivial solutions from this trivial state; and third, a Hopf bifurcation analysis locating the onset of oscillatory behavior and proving the existence of periodic solutions. Finally, we complement our theoretical results with numerical simulations that not only validate our predictions but also reveal complex, chaotic dynamics. 
These simulations highlight the rich pattern formation potential of the extended system and demonstrate behaviors not previously captured in earlier models. 

\section{Steady State Bifurcation Analysis}
In this section, we perform a bifurcation analysis to demonstrate the existence of non-constant steady-state solutions, thereby showing that the system can generate stable spatial patterns under suitable conditions.  
For the sake of simplicity, we shift the scalar fields in system \eqref{eq:E} using the steady state \eqref{eq:CSS} by setting 
\begin{equation*}
    \tilde A(\textbf{x}, t):=A(\textbf{x}, t)-\bar A, \quad 
    \tilde \rho(\textbf{x}, t):=\rho-\bar \rho, \quad 
    \tilde u(\textbf{x}, t):=u(\textbf{x}, t)-\bar u. 
\end{equation*}
and collect the following the equivalent system for $(\bar{A}, \bar{\rho}, \bar{u})$
\begin{equation*}\label{eq:ME}
\tag{ME}\left\{
\begin{aligned}
     \frac{\partial \tilde{A}}{\partial t}=&D_A\Delta \tilde{A} - \tilde{A}+\tilde{A}\tilde{\rho} +\rbar\tilde{A}+ \Abar\tilde{\rho}+\Abar\rbar+\alpha, &&\text{ in }\Omega\times(0, T], \\
     \frac{\partial \tilde{\rho}}{\partial t}=&\nabla\cdot(\nabla D_\rho \tilde{\rho} - \frac{2(\tilde{\rho}+\rbar)}{\tilde{A}+\Abar}\nabla \tilde{A})&&\\
     &- \tilde{A} \tilde{\rho}-\Abar \tilde{\rho}-\rbar\tilde{A}-\tilde{\rho}\tilde{u}-\ubar\tilde{\rho}-\rbar\tilde{u}+\beta-\rbar(\ubar+\Abar), &&\text{ in }\Omega\times(0, T], \\
    \frac{\partial \tilde{u}}{\partial t}=&\nabla\cdot(\nabla D_u \tilde{u} - \frac{\chi(\tilde{u}+\ubar)}{\tilde{A}+\Abar}\nabla \tilde{A}), &&\text{ in }\Omega\times(0, T], \\
    \frac{\partial \tilde A}{\partial \textbf{n}}=& \frac{\partial \tilde \rho}{\partial \textbf{n}}= \frac{\partial \tilde u}{\partial \textbf{n}}=0, &&\text{ on }\partial\Omega.
\end{aligned}
\right.
\end{equation*}
This new system \eqref{eq:ME} has a shifted equilibrium $\mathbf{0}:=(0, 0, 0)$, and  dropping tildes we obtain
\begin{equation}\label{eq:MP}
\tag{MP}
\left\{
\begin{aligned}
     \frac{\partial A}{\partial t}=&D_A\Delta A - A +A\rho +\left(1-\frac{\alpha}{\lambda}\right)A+\lambda \rho, &&\quad\text{ in }\Omega\times(0, T], \\
     \frac{\partial \rho}{\partial t}=&\nabla\cdot\Big(\nabla D_\rho \rho - \frac{2\left(\rho+\left(1-\frac{\alpha}{\lambda}\right)\right)}{A+\lambda}\nabla A\Big)- A \rho-\lambda \rho &&\\
     &-\left(1-\frac{\alpha}{\lambda}\right)A-\rho u-\rho\lambda\left(\frac{\beta}{\lambda-\alpha}-1\right)-u\left(1-\frac{\alpha}{\lambda}\right), &&\quad\text{ in }\Omega\times(0, T], \\
    \frac{\partial u}{\partial t}=&\nabla\cdot\Big(\nabla D_u u - \frac{\chi\left(u+\lambda\left(\frac{\beta}{\lambda-\alpha}-1\right)\right)}{A+\lambda}\nabla A\Big), &&\quad\text{ in }\Omega\times(0, T], \\
    \frac{\partial A}{\partial \textbf{n}}=& \frac{\partial \rho}{\partial \textbf{n}}= \frac{\partial u}{\partial \textbf{n}}=0, &&\quad\text{ on }\partial\Omega.
\end{aligned}
\right.
\end{equation}
We now aim to study the steady-state solutions of system \eqref{eq:MP}, and to this end, we introduce the operator $\mathcal F$ as follows
\begin{equation}\label{eq:F}
\begin{split}
    &\mathcal F(\chi, A, \rho, u)\\
    :=&\begin{bmatrix}
        D_A\Delta A+A\rho -\frac{\alpha}{\lambda}A+\lambda \rho\\
     \nabla\cdot\left[\nabla D_\rho \rho - \frac{2\left(\rho+\left(1-\frac{\alpha}{\lambda}\right)\right)}{A+\lambda}\nabla A\right]- A \rho-\lambda \rho-\left(1-\frac{\alpha}{\lambda}\right)A-\rho u-\rho\lambda\left(\frac{\beta}{\lambda-\alpha}-1\right)-u\left(1-\frac{\alpha}{\lambda}\right)\\
    \nabla\cdot\left[\nabla D_u u - \frac{\chi\left(u+\lambda\left(\frac{\beta}{\lambda-\alpha}-1\right)\right)}{A+\lambda}\nabla A\right]
    \end{bmatrix}.
\end{split}
\end{equation}

Then $\mathcal F(\chi, \mathbf{0})=0$ for any $\chi\in\mathbb{R}$, and we are motivated to find nontrivial equilibria of \eqref{eq:MP} that bifurcate from the trivial one.  It is equivalent to solving the following problem
\begin{equation}\label{eq:EP}\tag{EP}
\mathcal F(\chi, A, \rho, u)=0, \quad \chi \in \mathbb R,\quad (A, \rho, u)\in \mathcal Y, 
\end{equation}
where 
\begin{equation*}
   \mathcal Y=\left\{(A, \rho, u)\in[H^2(\Omega)]^3\Big|\;\frac{\partial A}{\partial \textbf{n}}=\frac{\partial \rho}{\partial \textbf{n}}=\frac{\partial u}{\partial \textbf{n}}=0\text{ on }\partial\Omega, \, \int_\Omega u d\textbf{x}=0\right\}.
\end{equation*}
Indeed, any root of $\mathcal{F} = 0$ is known to be a weak solution of \eqref{eq:F}, and standard elliptic regularity theory then ensures that this weak solution is, in fact, classical and smooth.

We will show that the conditions of Theorem 1.7 in \cite{crandall1971bifurcation} are satisfied for $\mathcal{F}=F$ in \eqref{eq:EP}, $\mathcal{Y}=Y$, $\mathcal{Z}=Z$, and $\mathcal{V}=V$ defined as follows

\begin{equation*}
    V=\left\{(\chi, A, \rho, u)\in\mathbb{R}\times \mathcal Y|\;
    A>-\lambda, \, \rho>\frac{\alpha}{\lambda}-1, \, u>\lambda \left(1-\frac{\beta}{\lambda-\alpha}\right)\right\}, 
\end{equation*}
and $Z:=[L^2(\Omega)]^3$.
\begin{remark}
Whenever we refer to Theorem 1.7 in \cite{crandall1971bifurcation}, we mean its extension given in Theorem 4.3 of \cite{shi2009global}. Part (e), though not stated explicitly in either Theorem, concerns the global extension, where one shows that the derivative is a Fredholm operator. For the complete statement of the extended Theorem, we refer the reader to the cited papers.
\end{remark}

Now that $\mathcal F(\chi, \textbf{0})=0$ holds for all $\chi\in\mathbb R$, to look for nonconstant solutions, one first needs the implicit function Theorem to fail at this trivial location.  For this purpose, we consider the following eigenvalue problem
\begin{equation}\label{eq:LN}\tag{LN}
\Delta \phi+\mu \phi =0 \text{ in }\Omega, \text{ with }
\partial_\textbf{n} \phi=0 \text{ on }\partial\Omega  \quad  \text{ and } \Vert \phi \Vert_{L^2(\Omega)}=1.  
\end{equation}
It is well known that its eigen-pairs $<\mu_j, \phi_j>$ consist of a sequence of real, non-negative, and diverging eigenvalues:
\[0 = \mu_0 < \mu_1 \leq \mu_2 \leq \dots \to \infty, \]
and the eigenfunction set \( \{\phi_j\}_{j=0}^\infty \subset H^1(\Omega) \) form a complete orthonormal basis of \( L^2(\Omega) \).  To apply the bifurcation from simple eigenfunctions, we assume that $\mu$ has algebraic multiplicity one.  Then $y_0$ in Theorem 1.7 in \cite{crandall1971bifurcation} is relabeled as $\textbf{P}(\textbf{x}) $ and defined as follows
\begin{equation}\label{Pvec}
\begin{split}
     \textbf{P}(\textbf{x}) =\left(1, \frac{1}{\lambda}\left( D_A\mu+\frac{\alpha}{\lambda}\right), \frac{(2\mu-\lambda)(\lambda-\alpha)^2-(D_A\lambda\mu+\alpha)(D_\rho\mu(\lambda-\alpha)+\beta\lambda)}{\lambda(\lambda-\alpha)^2}\right)\phi(\textbf{x}).
\end{split}
\end{equation}

The associated linear functional on $Y$ is then
\begin{equation*}
    \ell\varphi=\int_\Omega \varphi\cdot \textbf{P}(\textbf{x}), 
\end{equation*}
which allows us to define the closed complement $\mathcal{W}=W$
\begin{equation*}
    W=\{\varphi\in Y|\ell\varphi=0\}.
\end{equation*}
Using the definitions above and Theorem 1.7 in \cite{crandall1971bifurcation}, we extend Theorem 2.2 and Theorem 2.6 in \cite{cantrell2012global} as follows. 
\begin{theorem}\label{1stthm}
Let $\mu>0$ be a generic and simple eigenvalue of \eqref{eq:LN}.  Denote
\begin{equation}\label{eq:chisol}
    \chi=
    \overbrace{    -\frac{D_u}{(\beta+\alpha-\lambda)}\left[\left(\lambda-\alpha+\frac{\alpha\beta}{\lambda-\alpha}\right)+\mu\left(-2\left(1-\frac{\alpha}{\lambda}\right)+\frac{D_\rho \alpha}{\lambda}+\frac{D_A\beta\lambda}{\lambda-\alpha}\right)+\mu^2D_AD_\rho\right]}^{\triangleq \chi_0}
\end{equation}
and assume that 
\begin{align*}
    \frac{(D_u-\chi_0)(\alpha^2+\alpha(\beta-2\lambda)+\lambda^2)-\beta\lambda\chi_0}{\mu(\lambda-\alpha)}\
\end{align*}
is not eigenvalue of \eqref{eq:LN}.  Then a branch of spatially nonconstant solutions of $\mathcal F$ bifurcates from the equilibrium $(\bar{A}, \rbar, \bar{u})$ at $\chi=\chi_0$. In a neighborhood of the bifurcation point, the bifurcating branch can be parametrized as  
     $(A, \rho, u)(\textbf{x};s)=(\bar{A}, \rbar, \bar{u})+s\textbf{P}(\textbf{x})+s^2\textbf{Q}(\textbf{x};s)$ with $\textbf{P}$ given by \eqref{Pvec} and $\textbf{Q}(\textbf{x};s)\in W$ for any $s\in(-\epsilon, \epsilon)$; and where $\chi(0)=\chi_0$ and $\xi(0)=0$. Furthermore, the bifurcation branch is part of a connected component $\mathcal{C}$ of the set $\bar{\mathcal{S}}$, where $\mathcal{S}=\{(\chi, A, \rho, u):(\chi, A-\bar{A}, \rho-\rbar, u-\bar{u})\in V\, |\, \mathcal F(\chi, A-\bar{A}, \rho-\rbar, u-\bar{u})=0, (A, \rho, u)\not=(\bar{A}, \rbar, \bar{u})\}$, and $\mathcal{C}$ either is not compact or contains point $(\chi^*, \bar{A}, \rbar, \bar{u})$ with $\chi^*\not=\chi_0$.
\end{theorem}

\begin{remark}
In \cite{cantrell2012global}, the authors showed that if $\mathcal{C}$ is not compact, then it must be unbounded (see Remark 2.4). This implies that $\mathcal{C}$ either extends to infinity in one of the parameters $\chi$, $A$, $\rho$, or $u$, or else leaves $V$, in which case $\mathcal{C}$ contains a point on $\partial V$. This conclusion, established in their Lemma 2.5, relies on the theory developed in \cite{ladyzhenskaya1968}. The same lemma and conclusion can be applied here; the only distinction is that we must additionally show that $u$ is uniformly bounded in $C^{2,\vartheta}(\bar{\Omega})$. However, the proof of this step is identical to the argument used for $\rho$ in \cite{cantrell2012global}.
\end{remark} 

Mathematically, the value $\chi_0$ marks the point at which the homogeneous steady state loses stability and new nontrivial steady states emerge through bifurcation. From a criminological perspective, this corresponds to a tipping point in policing intensity: below $\chi_0$, police and offenders distribute uniformly in space; however, once $\chi_0$ is exceeded, stable, spatially inhomogeneous patterns arise, which may manifest as persistent crime hotspots despite ongoing policing efforts.

\subsection{Identification of bifurcation points}

To apply the Crandall--Rabinowitz bifurcation theory \cite{crandall1971bifurcation, shi2009global}, we first note that $\mathcal F$ is Fr\'echet differentiable and its functional derivative at $(A, \rho, u)$  with $\textbf{v}:=(v_1, v_2, v_3)$ is given by
\begin{equation}\label{eq:DFARU}
\begin{split}
&D\mathcal F_{(A, \rho, u)}(\chi, A, \rho, u)[\textbf{v}]\\
=&
\begin{bmatrix}
    D_A \Delta v_1 
      + \Bigl(\rho - \tfrac{\alpha}{\lambda}\Bigr)v_1  
      + (A+\lambda)\, v_2
    \\[1em]
    \nabla \cdot \Biggl[
        \nabla D_\rho v_2 
        - \tfrac{2v_2}{A+\lambda}\,\nabla A
        - 2\Bigl(\rho+1-\tfrac{\alpha}{\lambda}\Bigr)
          \left(
            \tfrac{\nabla v_1}{A+\lambda}
            - \tfrac{v_1 \nabla A}{(A+\lambda)^2}
          \right)
    \Biggr]
    \\[1em]
    \quad
    -\Bigl(\rho+1-\tfrac{\alpha}{\lambda}\Bigr)\, v_1
    - \left(A+u+\tfrac{\lambda\beta}{\lambda-\beta}\right) v_2
    - \Bigl(\rho+1-\tfrac{\alpha}{\lambda}\Bigr) v_3
    \\[1em]
    \nabla \cdot \Biggl[
        \nabla D_u v_3
        - \tfrac{\chi v_3}{A+\lambda}\,\nabla A
        - \chi\left(
            u + \lambda\Bigl(\tfrac{\beta}{\lambda-\alpha}-1\Bigr)
          \right)
          \left(
            \tfrac{\nabla v_1}{A+\lambda}
            - \tfrac{v_1 \nabla A}{(A+\lambda)^2}
          \right)
    \Biggr]
\end{bmatrix}.
\end{split}
\end{equation}
One can continue to show that $\frac{d}{d\chi}D\mathcal F_{(A, \rho, u)}(\chi, A, \rho, u)[\textbf{v}]$ exists and is continuous.

To fail the implicit function Theorem at $(\chi, \mathbf{0})$, we evaluate \eqref{eq:DFARU} the following way
\begin{equation}\label{eq:DFARU0}
    \begin{split}
        &D\mathcal F_{(A, \rho, u)}(\chi, \textbf{0})[\textbf{v}]\\
=&
    \begin{bmatrix}
        D_A\Delta v_1-\frac{\alpha}{\lambda}v_1+\lambda v_2\\[0.5em]
        D_\rho \Delta v_2-\frac{2}{\lambda}\left(1-\frac{\alpha}{\lambda}\right)\Delta v_1-\left(1-\frac{\alpha}{\lambda}\right)v_1-\frac{\lambda\beta}{\lambda-\alpha}v_2-\left(1-\frac{\alpha}{\lambda}\right)v_3\\[0.5em]
        D_u\Delta v_3-\chi \left(\frac{\beta}{\lambda-\alpha}-1\right)\Delta v_1
    \end{bmatrix}.
    \end{split}
\end{equation}

Then we are to look for nontrivial solutions of the following equation
\begin{equation*}\label{eq:LE}
    \tag{LE}
    \left\{
    \begin{aligned}
        &D_A\Delta v_1-\frac{\alpha}{\lambda}v_1+\lambda v_2=0, &&\quad\text{ in }\Omega, \\
        &D_\rho \Delta v_2-\frac{2}{\lambda}\left(1-\frac{\alpha}{\lambda}\right)\Delta v_1-\left(1-\frac{\alpha}{\lambda}\right)v_1-\frac{\lambda\beta}{\lambda-\alpha}v_2-\left(1-\frac{\alpha}{\lambda}\right)v_3=0, &&\quad\text{ in }\Omega, \\
        &D_u\Delta v_3-\chi \left(\frac{\beta}{\lambda-\alpha}-1\right)\Delta v_1=0, &&\quad\text{ in }\Omega, \\
         &\frac{\partial v_1}{\partial \textbf{n}}= \frac{\partial v_2}{\partial \textbf{n}}= \frac{\partial v_3}{\partial \textbf{n}}=0, &&\quad\text{ on }\partial\Omega.  
    \end{aligned}
    \right.
\end{equation*}

We test each equation with $\phi$ and then apply integration by parts to obtain
\begin{equation}\label{eq:Hmatrix}
\overbrace{
\begin{bmatrix}
        -D_A\mu -\frac{\alpha}{\lambda}&\lambda&0\\[0.5em]
        \left(\frac{2\mu}{\lambda}-1\right)\left(1-\frac{\alpha}{\lambda}\right)&-D_\rho\mu-\frac{\beta\lambda}{\lambda-\alpha}&-\left(1-\frac{\alpha}{\lambda}\right)\\[0.5em]
        \mu\chi\left(\frac{\beta}{\lambda-\alpha}-1\right)&0&-D_u\mu
    \end{bmatrix}
}^{\triangleq H}
    \begin{bmatrix}
        \int_\Omega v_1\phi \\
        \int_\Omega v_2\phi \\
        \int_\Omega v_3\phi
    \end{bmatrix}
    =    \begin{bmatrix}
        0 \\
        0 \\
        0
    \end{bmatrix}
    .
\end{equation}
Then the existence of nontrivial solutions of \eqref{eq:LE} will in turn give us nontrivial solutions for \eqref{eq:Hmatrix}.  That said, matrix $H$ must be non-invertible such that  
\[-D_u\mu\left(D_A\mu+\frac{\alpha}{\lambda}\right)\left(D_\rho\mu+\frac{\beta\lambda}{\lambda-\alpha}\right)+ D_u\lambda\mu\left(\frac{2\mu}{\lambda}-1\right)\left(1-\frac{\alpha}{\lambda}\right) -\lambda\mu\chi\left(1-\frac{\alpha}{\lambda}\right)\left(\frac{\beta}{\lambda-\alpha}-1\right)=0\]
or equivalently when $\chi$ is given by \eqref{eq:chisol}.

This analysis identifies the critical bifurcation values of $\chi$, representing the policing intensity, that mark the transition from homogeneous to patterned states, as determined by the linear stability analysis of the constant solution \eqref{eq:CSS}. Moreover, when $\chi = \chi_0$, it follows directly from \eqref{eq:Hmatrix} that the kernel of $D\mathcal{F}(\chi_0, \mathbf{0})$ is one-dimensional and is spanned by $\mathbf{P}$ defined in \eqref{Pvec}.

To proceed further, we verify Agmon's condition following \cite{cantrell2012global, shi2009global}.  Define the principal part of the elliptic operator $F(\chi, A, \rho, u)$ and $D\mathcal F_{(A, \rho, u)}(\chi, A, \rho, u)$
\begin{equation*}
    A_1(\chi, A, \rho, u)=\begin{bmatrix}
        D_A&0&0\\[0.5em]
        -2\frac{\rho+\left(1-\frac{\alpha}{\lambda}\right)}{A+\lambda}&D_\rho&0\\[0.5em]
        -2\frac{u+\lambda\left(\frac{\beta}{\lambda-\alpha}-1\right)}{A+\lambda}&0&D_\rho
    \end{bmatrix}.
\end{equation*}
More importantly, we need to look at the determinant  $A_1(\chi, A, \rho, u)+\sigma I$, $\sigma\in \mathbb{C}$
\begin{equation*}
    \begin{vmatrix}
        D_A+\sigma&0&0\\[0.5em]
        -2\frac{\rho+\left(1-\frac{\alpha}{\lambda}\right)}{A+\lambda}&D_\rho+\sigma&0\\[0.5em]
        -\chi\frac{u+\lambda\left(\frac{\beta}{\lambda-\alpha}-1\right)}{A+\lambda}&0&D_u+\sigma
    \end{vmatrix}=(D_A+\sigma)(D_\rho+\sigma)(D_u+\sigma).
\end{equation*}
The later is nonzero when $\sigma=0$ or $\arg \sigma\in[-\pi/2, \pi/2]$. This, in turn, implies that the operator $D\mathcal F_{(A, \rho, u)}(\chi, A, \rho, u)$ is a Fredholm operator with index 0 and $\text{Range}(D\mathcal F_{(A, \rho, u)}(\chi_0, \mathbf{0}))$ has a codimension of 1.

\subsection{Transversality and steady-state bifurcation}

We now verify that bifurcation does occur at $\chi_0$ given by \eqref{eq:chisol} by proving the following transversality condition 
\[\frac{d}{d\chi}D\mathcal F_{(A, \rho, u)}(\chi_0, \mathbf{0})[\textbf{P}(\textbf{x}) ]\not\in\\\text{Range}(D\mathcal F_{(A, \rho, u)}(\chi_0, \mathbf{0})).\] 
We argue by contradiction and assume that there exists nonzero $\mathbf{\hat v}:=(\hat v_1, \hat v_2, \hat v_3)$ such that $$D\mathcal F_{(A, \rho, u)} (\chi, \mathbf{0})[\mathbf{\hat v}]=\frac{d}{d\chi}D\mathcal F_{(A, \rho, u)}(\chi, \mathbf{0})[\textbf{P}(\textbf{x}) ].$$  Since
\begin{align*}\frac{d}{d\chi}D\mathcal F_{(A, \rho, u)}(\chi_0, \mathbf{0})[\textbf{P}(\textbf{x}) ]=\begin{bmatrix}
        0\\
        0\\
        \mu\left(\frac{\beta}{\lambda-\alpha}-1\right)\phi
    \end{bmatrix}, 
\end{align*}
we collect the following system
\begin{equation}\label{eq:S1}\tag{\text{$S_1$}} 
\left\{
\begin{aligned}
        &D_A\Delta \hat v_1-\frac{\alpha}{\lambda}\hat v_1+\lambda \hat v_2=0, &&\text{in }\Omega, \\
        &D_\rho \Delta \hat v_2-\frac{2}{\lambda}\left(1-\frac{\alpha}{\lambda}\right)\Delta \hat v_1-\left(1-\frac{\alpha}{\lambda}\right)\hat v_1-\frac{\lambda\beta}{\lambda-\alpha}\hat v_2-\left(1-\frac{\alpha}{\lambda}\right)\hat v_3
        =0, &&\text{in }\Omega, \\
        &D_u\Delta \hat v_3-\chi_0 \left(\frac{\beta}{\lambda-\alpha}-1\right)\Delta \hat v_1=\mu\left(\frac{\beta}{\lambda-\alpha}-1\right)\phi, &&\text{in }\Omega, \\
        &\frac{\partial \hat v_1}{\partial \textbf{n}}=\frac{\partial \hat v_2}{\partial \textbf{n}}=\frac{\partial \hat v_3}{\partial \textbf{n}}=0, &&\text{on }\partial \Omega.  
\end{aligned}
\right.
\end{equation}

We test \eqref{eq:S1} by $\phi$ and then collect the following identity using the fact that $\Vert \phi\Vert_{L^2}=1$ 
\begin{equation*} 
\begin{bmatrix}
        -D_A\mu -\frac{\alpha}{\lambda}&\lambda&0\\[0.5em]
        \left(\frac{2\mu}{\lambda}-1\right)\left(1-\frac{\alpha}{\lambda}\right)&-D_\rho\mu-\frac{\beta\lambda}{\lambda-\alpha}&-\left(1-\frac{\alpha}{\lambda}\right)\\[0.5em]
        \mu\chi_0\left(\frac{\beta}{\lambda-\alpha}-1\right)&0&-D_u\mu
    \end{bmatrix} 
    \begin{bmatrix}
        \int_\Omega \hat v_1\phi \\
        \int_\Omega \hat v_2\phi \\
        \int_\Omega \hat v_3\phi
    \end{bmatrix}
    =    \begin{bmatrix}
        0 \\
        0 \\
\mu\left(\frac{\beta}{\lambda-\alpha}-1\right)
    \end{bmatrix}
    .
\end{equation*}
However, this reaches a contradiction to \eqref{eq:Hmatrix} since $H$ therein is non-invertible with $\chi=\chi_0$.  This verifies the transversality condition, and Theorem \ref{1stthm} follows from those in \cite{crandall1971bifurcation}. 

\section{Stability Analysis of the Bifurcating Steady States}

In this section, we establish conditions for the stability of spatially heterogeneous patterns. Unlike previous analyses, our approach explicitly incorporates the geometry of the domain. Nevertheless, the results remain applicable to a broad class of domains, including discs. To this end, we introduce the following cubic equation in $r$, whose derivation will be presented later, but is needed here for the precise formulation of the Theorem:
\begin{equation}\label{eq:a2a1a0=0}
    \begin{split}
        0=&\biggl\{\left(\frac{\alpha}{\lambda}+\frac{\beta\lambda}{\lambda-\alpha}\right)\left[\left(\lambda-\alpha\right)+\frac{\alpha\beta}{\lambda-\alpha}\right]\biggr\}\\
        +&r\biggl\{\left(D_A+D_\rho+D_u\right)\left(\left(\lambda-\alpha\right)+\frac{\alpha\beta}{\lambda-\alpha}\right)\\
        &+\left(\frac{\alpha}{\lambda}+\frac{\beta\lambda}{\lambda-\alpha}\right)\left[-2\left(1-\frac{\alpha}{\lambda}\right)+\frac{D_\rho\alpha}{\lambda}+\frac{D_A\beta\lambda}{\lambda-\alpha}+D_u\left(\frac{\beta\lambda}{\lambda-\alpha}+\frac{\alpha}{\lambda}\right)\right]\\
        &+\mu\left[\left(-2\left(1-\frac{\alpha}{\lambda}\right)+\frac{ D_\rho\alpha}{\lambda}+\frac{D_A\beta\lambda}{\lambda-\alpha}\right)+\mu D_AD_\rho\right]\biggr\}\\
        +&r^2\biggl\{\left(D_A+D_\rho\right)\left[-2\left(1-\frac{\alpha}{\lambda}\right)+\frac{D_\rho\alpha}{\lambda}+\frac{D_A\beta\lambda}{\lambda-\alpha}\right]+\left[\left(D_A+D_\rho+D_u\right)\left(\frac{\beta\lambda}{\lambda-\alpha}+\frac{\alpha}{\lambda}\right)\right]\\
        &+\left(\frac{\alpha}{\lambda}+\frac{\beta\lambda}{\lambda-\alpha}\right)\left(D_AD_\rho+D_AD_u+D_\rho D_u\right)\biggr\}\\ +&r^3\biggl\{\left(D_A+D_\rho\right)\left(D_A+D_u\right)\left(D_\rho+D_u\right)\biggr\}\\
        =&C\left(r-r_1\right)\left(r-r_2\right)\left(r-r_3\right), 
    \end{split}
\end{equation}
where $r_i$ denote the roots of \eqref{eq:a2a1a0=0}. At least one of these roots must be negative; without loss of generality, we denote this root by $r_3$.

\begin{theorem}\label{2ndthm}
    Suppose that $\mu$ is the eigenvalue of $-\Delta$ appearing in Theorem \ref{1stthm} with normalized eigenfunction $\phi$. Suppose that the hypotheses of Theorem \ref{1stthm} are satisfied. Assume that $\int_\Omega\phi^3\not=0$ and $\beta\not=-C_2/C_1$, where
    \begin{equation*}
    \begin{split}
        C_1=- \lambda  \left(\chi _0-D_u\right) \left(-\alpha  \chi _0+2 D_A D_u\lambda  \mu  +2D_u \alpha  +\lambda  \chi _0\right), 
    \end{split}
\end{equation*}
\begin{equation*}
\begin{split}
        C_2=&-(\lambda -\alpha) \Biggl\{\lambda  \biggl[-2  D_u^2 \mu \Bigl(D_A (\mu  D_{\rho }-\lambda +\mu )-1\Bigr)+D_u\lambda  \chi_0  (1-2 \mu  D_A)-\lambda  \chi _0^2\biggr]\\
        &+\alpha  \biggl[2 D_u^2 \Bigl(\lambda -\mu  (D_{\rho}+2)\Bigr)-3D_u \lambda  \chi_0 +\lambda  \chi_0^2\biggr]\Biggr\}.
\end{split}
\end{equation*}
Suppose that there are no eigenvalue of $-\Delta$, $\nu$, such that $(\max\{0, \theta-\mu\}<\nu< \max\{\mu, \theta-\mu\})$, where
\begin{equation*}
\theta=\frac{2\lambda(\lambda-\alpha)-\alpha(\lambda-\alpha)(2+D_\rho)-D_A\beta\lambda^2}{D_AD_\rho\lambda(\lambda-\alpha)}.
\end{equation*}
If $r_1$ and $r_2$ are both positive roots of \eqref{eq:a2a1a0=0}, further suppose that, $\mu$ (and $\theta-\mu$ in case it is an eigenvalue of $-\Delta$) are outside of the interval $(\min\{r_1, r_2\}, \max\{r_1, r_2\})$.  Then the branch $(\chi(s), A(s), \rho(s), u(s))$ of solution of \eqref{eq:EQ} bifurcating from $(\chi_0, \Abar, \rbar, \ubar)$ is stable for sufficiently small and non-zero $|s|$. 
\end{theorem}

\begin{proof}

We proceed by following the approach outlined in \cite{cantrell2012global}, where the authors apply the bifurcation theory developed in \cite{crandall1973bifurcation}.  The stability of the bifurcation branch $(\chi(s), A(s), \rho(s), u(s))$ of \eqref{eq:EP} (or equivalently \eqref{eq:EQ}, by translation) is determined by the eigenvalues of $D\mathcal{F}{(A,\rho,u)}(\chi(s), A(s), \rho(s), u(s))$.
In particular, the branch is stable if all eigenvalues have negative real parts. Theorem \ref{1stthm} implies that
$D\mathcal{F}{(A,\rho,u)}(\chi_0, \mathbf{0})$
has $0$ as an $i$-simple eigenvalue, where $i$ is an injection. Moreover, Corollary 1.13 in \cite{crandall1973bifurcation} ensures the existence of intervals $I$ and $J$ with $\chi_0 \in I$ and $0 \in J$, together with continuously differentiable functions
$\gamma : I \to \mathbb{R}, \quad
\sigma : J \to \mathbb{R}, \quad
[\mathbf{v}] : I^3 \to Y$ and $
(w_1, w_2, w_3): J^3 \to Y$ such that 
\begin{equation*}
    D\mathcal F_{(A, \rho, u)}(\chi, \mathbf{0})[\textbf{v}]^T=\gamma(\chi)[\textbf{v}]^T, 
\end{equation*}
\begin{equation*}
    D\mathcal F_{(A, \rho, u)}(\chi(s), A(s), \rho(s), u(s))(w_1, w_2, w_3)^T=\sigma(s)(w_1, w_2, w_3)^T, 
\end{equation*}
such that 
\begin{equation*}
        \chi(0)=\chi_0, \quad\gamma(\chi_0)=\sigma(0)=0, \quad (v_1(\chi_0), v_2(\chi_0), v_3(\chi_0))=(w_1(0), w_2(0), w_3(0))=\textbf{P}(\textbf{x}), 
\end{equation*}
\begin{equation*}
            (v_1(\chi), v_2(\chi), v_3(\chi))-\textbf{P}(\textbf{x}) \in W, \quad(w_1(s), w_2(s), v_3(s))-\textbf{P}(\textbf{x}) \in W.
\end{equation*}

Theorem 1.16 in \cite{crandall1973bifurcation} states that under the assumptions and definitions of Corollary 1.13, $\sigma(s)$ has only zero at $s=0$ and the same sign as $s$ for $|s|>0$ for sufficiently small $|s|$ and whenever $\chi'(0)\not=0$. It follows that if $\gamma(\chi_0)=\sigma(0)$ is the eigenvalue of \eqref{eq:F} with the largest real part  at $(\chi_0, \mathbf{0})$, then again for sufficiently small nonzero $s$, $\sigma(s)$ is  eigenvalue of the linearization of \eqref{eq:F} at $(\chi(s), A(s), \rho(s), u(s))$ with the largest real part. In particular, if $\mathcal{R}e(\sigma(0))<0$, then $\mathcal{R}e(\sigma(s))<0$. 

Hence, it remains to show that, under assumptions of Theorem \ref{2ndthm}, (a) $\chi'(0) \neq 0$, and (b) $\sigma(0) = 0$ is the largest eigenvalue of the linearized operator $D\mathcal{F}_{(A,\rho,u)}(\chi,A,\rho,u)$ at $(\chi_0,\mathbf{0})$.

To this end, we write system \eqref{eq:EP} with the bifurcating solution as
\[\mathcal{F}(\chi(s), A(s), \rho(s), u(s)) = 0,\]
with $\mathbf{x}$ omitted for notational simplicity.  We then differentiate twice with respect to $s$ at $s = 0$ to obtain. 
\begin{equation*}
    \begin{split}
        D_A\Delta A''(0)+2A'(0)\rho'(0)-\frac{\alpha}{\lambda}A''(0)+\lambda \rho''(0)=0,
    \end{split}
\end{equation*} 
\begin{equation*}
    \begin{split}
        &D_\rho\Delta \rho''\left(0\right)-\frac{2}{\lambda}\left(1-\frac{\alpha}{\lambda}\right)\Delta A''\left(0\right)-\nabla\cdot\left[\frac{4}{\lambda}\rho'\left(0\right)\nabla A'\left(0\right)\right]\\
    &\quad +\nabla\cdot\left[\frac{4}{\lambda^2}\left(1-\frac{\alpha}{\lambda}\right)A'\left(0\right)\nabla A'\left(0\right)\right]-2A'\left(0\right)\rho'\left(0\right)-\left(1-\frac{\alpha}{\lambda}\right)A''\left(0\right)\\
    &\quad-2\rho'\left(0\right)u'\left(0\right)-\frac{\lambda\beta}{\lambda-\beta}\rho''\left(0\right)-\left(1-\frac{\alpha}{\lambda}\right)u''\left(0\right)=0,
    \end{split}
\end{equation*}
and
\begin{equation*}
    \begin{split}
     &D_u\Delta u''\left(0\right)-\chi_0\left(\frac{\beta}{\lambda-\alpha}-1\right)\Delta A''\left(0\right)-\nabla \cdot \left[2\chi'\left(0\right)\left(\frac{\beta}{\lambda-\alpha}-1\right)\nabla A'\left(0\right)\right]\\
    &\quad-\nabla\cdot\left[2\chi_0\frac{1}{\lambda}u'\left(0\right)\nabla A'\left(0\right)\right]+\nabla\cdot\left[2\chi_0\frac{1}{\lambda}\left(\frac{\beta}{\lambda-\alpha}-1\right)A'\left(0\right)\nabla A'\left(0\right)\right]=0.   
    \end{split}
\end{equation*}

Multiply the above three equations by $\phi$ and integrate, and we should get the following
\begin{equation}\label{eq:Hint}
    H\left(\int_\Omega\phi A''(0), \int_\Omega\phi \rho''(0), \int_\Omega\phi u''(0)\right)^T+\widehat{h}=0, 
\end{equation}
where $\widehat{h}$ is the vector containing the lower derivative terms
\begin{equation*}
    \widehat{h}=\begin{bmatrix}
        \int_\Omega 2A'\left(0\right)\rho'\left(0\right)\phi\\[1em]
        \biggl\{\int_\Omega \biggl(-\nabla\cdot\left[\frac{4}{\lambda}\rho'\left(0\right)\nabla A'\left(0\right)\right] +\nabla\cdot\left[\frac{4}{\lambda^2}\left(1-\frac{\alpha}{\lambda}\right)A'\left(0\right)\nabla A'\left(0\right)\right]\\\quad-2A'\left(0\right)\rho'\left(0\right)-2\rho'\left(0\right)u'\left(0\right)\biggr)\phi\biggr\}\\[1em]
        \biggl\{\int_\Omega \biggl(-\nabla \cdot \left[2\chi'\left(0\right)\left(\frac{\beta}{\lambda-\alpha}-1\right)\nabla A'\left(0\right)\right]
    -\nabla\cdot\left[2\chi_0\frac{1}{\lambda}u'\left(0\right)\nabla A'\left(0\right)\right]\\
    \quad+\nabla\cdot\left[2\chi_0\frac{1}{\lambda}\left(\frac{\beta}{\lambda-\alpha}-1\right)A'\left(0\right)\nabla A'\left(0\right)\right]\biggr)\phi\biggr\}
    \end{bmatrix}.
\end{equation*}
We know that $\chi_0$ makes $H$ singular. The left eigenvector of $H$ corresponding to 0 is
\begin{equation*}
    \widehat{\textbf{P}(\textbf{x}) }=\left(-\frac{D_u\mu(\beta\lambda+ D_\rho\mu(\lambda-\alpha))}{(\lambda-\alpha)^2}, -\frac{D_u\lambda\mu}{\lambda-\alpha}, 1\right)^T.
\end{equation*}
We multiply \eqref{eq:Hint} by $\widehat{\textbf{P}(\textbf{x}) }$ from the left
\begin{equation*}
\begin{split}
        \widehat{\textbf{P}(\textbf{x}) }^TH\left(\int_\Omega\phi A''(0), \int_\Omega\phi \rho''(0), \int_\Omega\phi u''(0)\right)^T+\widehat{\textbf{P}(\textbf{x}) }^T\widehat{h}=\widehat{\textbf{P}(\textbf{x}) }^T0\implies \widehat{\textbf{P}(\textbf{x}) }^T\widehat{h}=0.
\end{split}
\end{equation*}
 We expand $\widehat{\textbf{P}(\textbf{x}) }^T\widehat{h}$ and use fact that $(A'(0), \rho'(0), u'(0))=\textbf{P}(\textbf{x}) $. In addition, integration by parts and Green's first identity let us deal with gradients in the following way
\begin{equation*}
    \int_\Omega\phi(\nabla\phi\cdot\nabla\phi)=\frac{1}{2}\int_\Omega\nabla\phi \cdot \nabla(\phi^2)=-\frac{1}{2}\int_\Omega\phi^2\Delta\phi=\frac{\mu}{2}\int_\Omega\phi^3.
\end{equation*}
We collect terms containing $\chi'(0)$ on the left-hand side (where we also use $\int_\Omega\phi^2=1$) and those containing $\int_\Omega\phi^3$ on the right-hand side
\begin{equation*}
    2D_u\lambda^2(\lambda-\alpha)(\alpha+\beta-\lambda)\chi'(0)=(C_1\beta+C_2)\int_\Omega\phi^3, 
\end{equation*}
where
\begin{equation}\label{C1}
    \begin{split}
        C_1:=- \lambda  \left(\chi _0-D_u\right) \left(-\alpha  \chi _0+2 D_A D_u\lambda  \mu  +2D_u \alpha  +\lambda  \chi _0\right), 
    \end{split}
\end{equation}
and
\begin{equation}\label{C2}
\begin{split}
        C_2=&-(\lambda -\alpha) \Biggl\{\lambda  \biggl[-2  D_u^2 \mu \Bigl(D_A (\mu  D_{\rho }-\lambda +\mu )-1\Bigr)+D_u\lambda  \chi_0  (1-2 \mu  D_A)-\lambda  \chi _0^2\biggr]\\
        &+\alpha  \biggl[2 D_u^2 \Bigl(\lambda -\mu  (D_{\rho}+2)\Bigr)-3D_u \lambda  \chi_0 +\lambda  \chi_0^2\biggr]\Biggr\}.
\end{split}
\end{equation}
The coefficients on the left-hand side are non-zero by assumption on the parameters. Hence, $\chi'(0)\not=0$ if $\beta\not=-C_2/C_1$ and $\int_\Omega\phi^3\not=0$.  This completes part (a).

\begin{remark}
Note that $\chi_0$ depends on $\beta$ through its definition. Thus, determining $\beta^*$ such that $C_2 \beta^* + C_1 \neq 0$ and ensuring that $\beta$ is independent of $\chi_0$ reduces to solving a quadratic equation. The explicit expression for $\beta^*$ can be obtained from the discriminant; however, it is algebraically cumbersome and therefore omitted here. Moreover, since $\beta > \lambda - \alpha > 0$, one can derive sufficient conditions to guarantee either $-C_2/C_1 < \lambda - \alpha$ or $-C_2/C_1 < 0$.
\end{remark}
\begin{remark}
As noted in \cite{cantrell2012global}, on a disc domain, a transcritical bifurcation arises because the eigenfunctions of the Laplace operator involve Bessel functions, which lack the symmetry required to ensure $\int_\Omega \phi^3 = 0$.  Consequently, the analysis developed here remains valid in this setting. In contrast, for rectangular domains the solutions of the Helmholtz equation are products of cosine functions, which satisfy $\int_\Omega \phi^3 = 0$ (for example, $\int_{[0,\pi]\times[0,\pi]} \cos^3(nx)\cos^3(my) dxdy = 0$).  Therefore, the same bifurcation scenario does not directly apply in rectangular domains.  This observation is consistent with previous studies by \cite{short2010nonlinear, yerlanov2025}, which demonstrate that pitchfork bifurcations are expected in rectangles, where $\chi'(0) = 0$. To analyze the stability of these bifurcating branches, one must therefore examine higher-order terms, in particular $\chi''(0)$.
\end{remark}

We follow \cite{cantrell2012global} to verify part (b): we show that $\sigma(0)=0$ is the eigenvalue of $D\mathcal F_{(A, \rho, u)}(\chi, A, \rho, u)$ at $(\chi_0, \mathbf{0})$ with the largest real part. The corresponding eigenvalue problem is
\begin{equation*}
\left\{
    \begin{aligned}
        &D_A\Delta v_1-\frac{\alpha}{\lambda}v_1+\lambda v_2=\sigma v_1, &&\quad\text{ in }\Omega, \\
        &D_\rho \Delta v_2-\frac{2}{\lambda}\left(1-\frac{\alpha}{\lambda}\right)\Delta v_1-\left(1-\frac{\alpha}{\lambda}\right)v_1-\frac{\lambda\beta}{\lambda-\alpha}v_2-\left(1-\frac{\alpha}{\lambda}\right)v_3=\sigma v_2, &&\quad\text{ in }\Omega, \\
        &D_u\Delta v_3-\chi \left(\frac{\beta}{\lambda-\alpha}-1\right)\Delta v_1=\sigma v_3, &&\quad\text{ in }\Omega, \\
         &\frac{\partial v_1}{\partial \textbf{n}}= \frac{\partial v_2}{\partial \textbf{n}}= \frac{\partial v_3}{\partial \textbf{n}}=0, &&\quad\text{ on }\partial\Omega, 
    \end{aligned}
    \right.
\end{equation*}
which can be rewritten as follows for $3\times3$ matrices $P, \, Q$
\begin{equation}\label{eq:SigmaMat}
    P\Delta\vec{v}+Q\vec{v}=\sigma\vec{v}.
\end{equation}
By spectral Theorem, $\vec{v}=\sum_{i=1}^\infty \vec{v}_i \phi_i$, where $\phi_i$ is orthonormal eigenfunction of $-\Delta$ with homogeneous Neumann boundary conditions. If $\sigma$ is an eigenvalue, then for some $i$, $\vec{v}_i\not=0$. Integrate \eqref{eq:SigmaMat} by $\phi_i$, then we should get
\begin{equation*}
    -P\mu_i\vec{v}_i+Q\vec{v}_i=\sigma v_i.
\end{equation*}
Hence, $\sigma$ can be represented as an eigenvalue of $-P\mu_i +Q$ for some eigenvalue $\mu_i$ of $-\Delta$. On the other hand, any $\mu_i$ leads to 3 eigenvalues of \eqref{eq:SigmaMat} corresponding to the eigenvalues of $-P\mu_i+Q$. To avoid confusion, we denote the eigenvalue of $-\Delta$ that appears in Theorems \ref{1stthm}, \ref{2ndthm} as $\mu_*$. We drop the $i$ subscript, then at the bifurcation point, $-P\mu+Q=H$ with $\chi=\chi_0$. We analyze the eigenvalues from the following equation
\begin{equation*}
    H\left(\int_\Omega v_1\phi, \int_\Omega v_2\phi, \int_\Omega v_3\phi\right)^T=\sigma\left(\int_\Omega v_1\phi, \int_\Omega v_2\phi, \int_\Omega v_3\phi\right)^T.
\end{equation*}
The characteristic equation is given as:
\begin{equation}\label{eq:charsigma}
    \sigma^3+a_2\sigma^2+a_1\sigma^1+a_0=0, 
\end{equation}
where
\begin{equation*}
    \begin{split}
    a_2:=&\frac{\alpha}{\lambda}+\frac{\beta\lambda}{\lambda-\alpha}+\mu\left(D_A+D_\rho+D_u\right), \\
    a_1:=&\left(\lambda-\alpha\right)+\frac{\alpha\beta}{\lambda-\alpha}+\mu\left[-2\left(1-\frac{\alpha}{\lambda}\right)+\frac{ D_\rho\alpha}{\lambda}+\frac{D_A\beta\lambda}{\lambda-\alpha}+D_u\left(\frac{\beta\lambda}{\lambda-\alpha}+\frac{\alpha}{\lambda}\right)\right]\\
    &\quad+\mu^2\left(D_AD_\rho+D_AD_u+D_\rho D_u\right), \\
    a_0:=&D_u\mu\biggl\{-\left[\mu_*\left(-2\left(1-\frac{\alpha}{\lambda}\right)+\frac{ D_\rho\alpha}{\lambda}+\frac{D_A\beta\lambda}{\lambda-\alpha}\right)+\mu_*^2D_AD_\rho\right]\\
    &\quad+\mu\left[-2\left(1-\frac{\alpha}{\lambda}\right)+\frac{ D_\rho\alpha}{\lambda}+\frac{D_A\beta\lambda}{\lambda-\alpha}\right]
    +\mu^2D_AD_\rho \biggr\}.
    \end{split}
\end{equation*}

According to Routh--Hurwitz stability criterion, \eqref{eq:charsigma} has roots with non-positive real parts if $a_2 \geq 0$, $a_0 \geq 0$, and $a_2 a_1 \geq a_0$. Our goal is to determine conditions that ensure none of the eigenvalues of $-\Delta$ violate these inequalities. Each coefficient $a_i$ is treated as a polynomial in $\mu$, and the corresponding equalities must be analyzed carefully. Note that $a_2 \geq 0$ for all non-negative values of $\mu$.

We focus on $a_0=0$.  We already know that, when condition $a_0(\chi,\mu)=0$ holds first at $\mu=\mu_1$, the first Neumann mode bifurcates; when it hits at $\mu_k$, mode $k$ bifurcates, leading to $k$-peak steady state patterns.  At the bifurcation point $(\chi_0, \mathbf{0})$ with $\mu=\mu_*$, one of the roots of \eqref{eq:charsigma} is $\sigma=0$, $a_0$ must be $0$. We get that the other two roots (besides $\mu_*$) are
\begin{equation*}
    \mu_\circ=0, \quad \mu_\dagger=\frac{1}{D_A D_\rho}\left[2\left(1-\frac{\alpha}{\lambda}\right)-\frac{\alpha D_\rho}{\lambda}-\frac{D_A\beta\lambda}{\lambda-\alpha}\right]-\mu_*, 
\end{equation*}
where the second one is obtained via Vieta's formula. We need to assume that there are no other eigenvalues of $-\Delta$ in the interval $[\max\{\mu_\circ, \mu_\dagger\}, \max\{\mu_*, \mu_\dagger\}]$ besides $\mu_\circ$, $\mu_*$, and $\mu_\dagger$.

We need to guarantee that for $\mu_*$, $\mu_\circ$, and $\mu_\dagger$ (in the case it is non-negative and an eigenvalue of $-\Delta$), the other two roots of \eqref{eq:charsigma} have a negative real part. We do so by looking at $a_2a_1-a_0$, which has the following form

\begin{equation}\label{eq:a2a1a0}
    \begin{split}
        &\mu^0\biggl\{\left(\frac{\alpha}{\lambda}+\frac{\beta\lambda}{\lambda-\alpha}\right)\left[\left(\lambda-\alpha\right)+\frac{\alpha\beta}{\lambda-\alpha}\right]\biggr\}\\
        +&\mu^1\biggl\{\left(D_A+D_\rho+D_u\right)\left(\left(\lambda-\alpha\right)+\frac{\alpha\beta}{\lambda-\alpha}\right)\\
        &+\left(\frac{\alpha}{\lambda}+\frac{\beta\lambda}{\lambda-\alpha}\right)\left[-2\left(1-\frac{\alpha}{\lambda}\right)+\frac{D_\rho\alpha}{\lambda}+\frac{D_A\beta\lambda}{\lambda-\alpha}+D_u\left(\frac{\beta\lambda}{\lambda-\alpha}+\frac{\alpha}{\lambda}\right)\right]\\
        &+\mu_*\left[\left(-2\left(1-\frac{\alpha}{\lambda}\right)+\frac{ D_\rho\alpha}{\lambda}+\frac{D_A\beta\lambda}{\lambda-\alpha}\right)+\mu_* D_AD_\rho\right]\biggr\}\\
        +&\mu^2\biggl\{\left(D_A+D_\rho\right)\left[-2\left(1-\frac{\alpha}{\lambda}\right)+\frac{D_\rho\alpha}{\lambda}+\frac{D_A\beta\lambda}{\lambda-\alpha}\right]+\left[\left(D_A+D_\rho+D_u\right)\left(\frac{\beta\lambda}{\lambda-\alpha}+\frac{\alpha}{\lambda}\right)\right]\\
        &+\left(\frac{\alpha}{\lambda}+\frac{\beta\lambda}{\lambda-\alpha}\right)\left(D_AD_\rho+D_AD_u+D_\rho D_u\right)\biggr\}\\ +&\mu^3\biggl\{\left(D_A+D_\rho\right)\left(D_A+D_u\right)\left(D_\rho+D_u\right)\biggr\}.
    \end{split}
\end{equation}
The coefficients in front of $\mu^0$ and $\mu^3$ are positive, hence by Descartes' rule of signs, we expect at least one negative solution. The other two roots can be both positive, negative, or complex conjugates. If they are negative or complex, then for any $\mu\geq0$, the above polynomial will be positive. If the roots are positive, say $\mu_1<\mu_2$, we need to put extra conditions that there are no eigenvalues of $-\Delta$ in $(\mu_1, \mu_2)$. This guarantees that there is no $\mu$ for which an associated $\sigma$ has a positive real part, including $\mu_\circ$, $\mu_*$, and $\mu_\dagger$. This concludes part (b).

\end{proof}

\begin{corol}\label{1stcor}
     Suppose that $\mu$ is the eigenvalue of $-\Delta$ appearing in Theorem \ref{1stthm} with normalized eigenfunction $\phi$. Suppose that the hypotheses of Theorem \ref{1stthm} are satisfied. Assume that $\int_\Omega\phi^3\not=0$ and $\beta\not=-C_2/C_1$, where $C_1$ and $C_2$ are given by \eqref{C1} and \eqref{C2}

Further suppose that 
\begin{equation*}
    -2\left(1-\frac{\alpha}{\lambda}\right)+\frac{ D_\rho\alpha}{\lambda}+\frac{D_A\beta\lambda}{\lambda-\alpha}>0.
\end{equation*}
If $\mu$ is the first positive eigenvalue of $-\Delta$, then the branch $(\chi(s), A(s), \rho(s), u(s))$ of solution of \eqref{eq:EQ} bifurcating from $(\chi_0, \Abar, \rbar, \ubar)$ is stable for sufficiently small and non-zero $|s|$. 
\end{corol}
\begin{proof}

With the same notation as in the proof of Theorem \ref{2ndthm}, suppose that
\begin{equation*}
    -2\left(1-\frac{\alpha}{\lambda}\right)
    +\frac{D_\rho \alpha}{\lambda}
    +\frac{D_A \beta \lambda}{\lambda-\alpha}
    \geq 0.
\end{equation*}
In this case, the coefficients of $\mu^1$ and $\mu^2$ in \eqref{eq:a2a1a0} are positive, which implies that no positive roots exist. However, this condition also forces $\mu_\dagger$ to be negative. Consequently, stability requires $\mu_*$ to be the first eigenvalue of $-\Delta$, since we must ensure that no eigenvalues lie in the interval $(0,\mu_*)$.
    
\end{proof}

With Corollary \ref{1stcor}, we arrive at a dichotomy. If
\begin{equation}\label{eq:mudag_1}
    -2\left(1-\frac{\alpha}{\lambda}\right)+\frac{ D_\rho\alpha}{\lambda}+\frac{D_A\beta\lambda}{\lambda-\alpha}>0,
\end{equation}
then stability is guaranteed only when $\mu$ corresponds to the first positive eigenvalue. On the other hand, if
\begin{equation}\label{eq:mudag_2}
    -2\left(1-\frac{\alpha}{\lambda}\right)+\frac{ D_\rho\alpha}{\lambda}+\frac{D_A\beta\lambda}{\lambda-\alpha}<-\mu D_AD_\rho<0,
\end{equation}
then $\mu_\dagger>0$, and stability requires that no eigenvalues of $\Delta$ fall within the interval $(\min\{\mu_, \mu_\dagger\},\\ \max\{\mu_, \mu_\dagger\})$ (recall, $\mu=\mu^*$).  This condition is less restrictive on $\mu_*$, since it does not force it to be the first eigenvalue. However, in this case, one must further ensure that no eigenvalues are present between $r_1$ and $r_2$, the positive solutions of \eqref{eq:a2a1a0=0}, if such exist.

These inequalities, and the resulting stability regions, depend sensitively on parameter choices. For instance, if $D_\rho>2\frac{\lambda-\alpha}{\alpha}$, then \eqref{eq:mudag_1} holds, while if $D_\rho<2\frac{\lambda-\alpha}{\alpha}$ and $D_A<\frac{(-D_\rho\alpha+2(\lambda-\alpha))(\lambda-\alpha)}{\lambda(\beta+D_\rho\mu(\lambda-\alpha))}$, then \eqref{eq:mudag_2} applies. Mathematically, this shows that smaller diffusion rates broaden the range of eigenvalues that can support stable non-constant branches, i.e., patterned steady states.

The stability of the uniform steady state depends very sensitively on the diffusion parameters of the model. Mathematically, the key inequalities show that when offender diffusion ($D_\rho$) and attractiveness diffusion ($D_A$) are relatively large, the system tends to remain close to uniform, and only the simplest spatial modes can be stable. By contrast, when both diffusion rates are small, the conditions for stability relax, and more spatial modes are permitted. In this regime, the uniform steady state can lose stability, and the system may bifurcate to non-uniform solutions—steady patterns that represent spatial clustering.

From a criminological perspective, these bifurcation points can be interpreted as thresholds in offender and guardian mobility. When offenders move broadly and randomly across space and when the attractiveness of a crime event spreads widely, local concentrations of crime are smoothed out. However, if offender movement is limited and the attractiveness effect remains strongly localized, the system can cross a threshold where uniform crime distributions are no longer sustainable. Instead, stable heterogeneous patterns emerge, corresponding to persistent hotspots.  Social science research suggests that such mechanisms may underlie the persistence of crime in certain neighborhoods. Limited offender mobility, coupled with strong near-repeat effects, reinforces clustering and creates conditions where particular areas remain chronically affected. In this sense, the mathematics formalizes a well-known observation in criminology: that reduced movement and strong local feedback can lock specific places into cycles of concentrated criminal activity. While empirical confirmation in the exact setting of this model is still lacking, the analysis provides a plausible theoretical explanation for the formation and persistence of crime hotspots. 

\section{Hopf Bifurcations and Time-Periodic Dynamics}

In this section, we investigate Hopf bifurcations, which correspond to the emergence of periodic solutions. Following a similar approach as in the previous section, we demonstrate that, under appropriate conditions, these cycles exist. We further validate the theoretical findings through numerical simulations of the original system. To begin, we introduce a key polynomial, which serves as the characteristic equation of $H$
\begin{equation}\label{eq:charH}
    \sigma^3+a_2(\mu)\sigma^2+a_1(\mu)\sigma^1+a_0(\mu)=0, 
\end{equation}
where
\begin{equation}\label{eq:a_iH}
    \begin{split}
    a_2\left(\mu\right)=&\frac{\alpha}{\lambda}+\frac{\beta\lambda}{\lambda-\alpha}+\mu\left(D_A+D_\rho+D_u\right), \\
    a_1\left(\mu\right)=&\left(\lambda-\alpha\right)+\frac{\alpha\beta}{\lambda-\alpha}+\mu\left[-2\left(1-\frac{\alpha}{\lambda}\right)+\frac{ D_\rho\alpha}{\lambda}+\frac{D_A\beta\lambda}{\lambda-\alpha}+D_u\left(\frac{\beta\lambda}{\lambda-\alpha}+\frac{\alpha}{\lambda}\right)\right]\\
    &\quad+\mu^2\left(D_AD_\rho+D_AD_u+D_\rho D_u\right), \\
    a_0\left(\chi, \mu\right)=&D_u\mu\biggl\{\left[\left(\lambda-\alpha\right)+\frac{\alpha\beta}{\lambda-\alpha}+\chi\left(\lambda-\alpha\right)\left(\frac{\beta}{\lambda-\alpha}-1\right)\right]\\
    &\quad+\mu\left[-2\left(1-\frac{\alpha}{\lambda}\right)+\frac{ D_\rho\alpha}{\lambda}+\frac{D_A\beta\lambda}{\lambda-\alpha}\right]
    +\mu^2D_AD_\rho\biggr\}.
    \end{split}
\end{equation}
We further need the following rational function
\begin{equation}\label{eq:rat}
\chi_*:=\frac{1}{\beta+\alpha-\lambda}\left\{\frac{a_1a_2}{D_u\mu}-\left(\lambda-\alpha\right)-\frac{\alpha\beta}{\lambda-\alpha}-\mu\left[-2\left(1-\frac{\alpha}{\lambda}\right)+\frac{ D_\rho\alpha}{\lambda}+\frac{D_A\beta\lambda}{\lambda-\alpha}\right]
-\mu^2D_AD_\rho\right\}.
\end{equation}
Note that the Theorem is stated for the shift problem \eqref{eq:EP}, but it holds for the original system \eqref{eq:E}.
\begin{theorem}\label{3rdthm}
 Let $\mu$ be a positive eigenvalue of $-\Delta$ with homogeneous Neumann boundary condition such that $a_1(\mu)>0$ in \eqref{eq:a_iH}, and $\mu$ be larger than local maximizer of $\chi_*$ in \eqref{eq:rat}. Then system \eqref{eq:EP} has in a neighborhood of $(\chi_*, \mathbf{0})\in V$ a unique one-parameter family $\gamma(s);0<s<\varepsilon$ of noncritical periodic orbits. More precisely: there exists $\varepsilon>0$ and 
    \begin{align*}
(\chi(\cdot), A(\cdot), \rho(\cdot), u(\cdot), T(\cdot))\in C^\infty((-\varepsilon, \varepsilon), V\times \mathbb{R}^+)
    \end{align*}
    with
    \begin{align*}
    (\chi(0), A(0), \rho(0), u(0), T(0))=\left(\chi_*, \Abar, \rbar, \ubar, \frac{2\pi}{\sqrt{a_1(\mu)}}\right)
    \end{align*}
    such that 
    \begin{align*}
        \gamma(s):=\gamma(A(s), \rho(s), u(s)), \quad 0<s<|\varepsilon|
    \end{align*}
    is a noncritical periodic orbit of \eqref{eq:EP} with $\chi=\chi(s)$ of period $T(s)$ passing through $A(s), \rho(s), u(s)\in Y$.  Moreover, the family $\{\gamma(s); 0<s<\varepsilon\}$ contains every noncritical periodic orbit of \eqref{eq:EP} lying in a suitable neighborhood of $(\chi_*, \mathbf{0}, T(0))\in V\times \mathbb{R}^+$, and $\gamma(s_1)\not=\gamma(s_2)$ for $s_1\neq s_2$
\end{theorem}
\begin{proof}

We verify that the conditions of Theorem 1 in \cite{amann2006hopf} are satisfied. In addition, we follow the proof strategies developed in \cite{rodriguez2021understanding, wang2017time}. Specifically, we analyze the eigenvalues of \eqref{eq:DFARU0}, the linearization of $\mathcal{F}$ at $\mathbf{0}$. The procedure mirrors our earlier approach: multiply by $\phi$ and integrate over the domain. This again reduces the problem to studying the characteristic equation of $H$, defined in \eqref{eq:charH} and \eqref{eq:a_iH}, and determining the conditions under which purely imaginary roots arise.

Let $\sigma_2$ and $\sigma_3$ be purely imaginary roots of the form $\pm i \omega_0$ with $\omega_0 > 0$, and let $\sigma_1$ denote the remaining real root. Then \eqref{eq:charH} can be rewritten as
\[(\sigma - \sigma_1)(\sigma^2 + \omega_0^2)
= \sigma^3 - \sigma_1 \sigma^2 + \omega_0^2 \sigma - \sigma_1 \omega_0^2.\]
It follows that $\sigma_1 = -a_2(\mu) < 0$ and $\sigma_{2,3} = \pm i \sqrt{a_1(\mu)}$. 
To guarantee that the roots are purely imaginary, we require $a_1(\mu) > 0$, which holds whenever $\mu$ is chosen outside the real roots of $a_1(\mu)$.  The critical bifurcation value $\chi_*(\mu)$ is then defined by the condition $a_2(\mu) a_1(\mu) = a_0(\mu)$.  In general, the rational function of $\mu$ given in \eqref{eq:rat} is continuous but not necessarily injective. 
Let $\mu_M$ denote a real local maximizer of this function.  Then, for any $\chi_*(\mu) > \chi_*(\mu_M)$, the corresponding value of $\mu$ is uniquely determined.  

Finally, we need to check the transversality condition. Let $\sigma_1(\chi)$ and $\sigma_{2, 3}(\chi)=p(\chi)+iq(\chi)$ be solutions to \eqref{eq:charH} in the neighborhood of $\chi=\chi_*$. Then, upon substitution and equating the coefficient in front of the power of $\sigma$, we get
\begin{align*}
    \left\{
    \begin{aligned}
        -a_2=&2p(\chi)+\sigma_1(\chi), \\
        a_1=&p^2(\chi)+q^2(\chi)+2p(\chi)\sigma_1(\chi), \\
        -a_0=&(p^2(\chi)+q^2(\chi))\sigma_1(\chi).
    \end{aligned}
    \right.
\end{align*}

Differentiating with respect to $\chi$ we get that $a_0(\mu)'=D_u\mu(\lambda-\alpha)\left(\frac{\beta}{\lambda-\alpha}-1\right)>0$.
\begin{align*}
    \left\{
    \begin{aligned}
        0=&2p'(\chi)+\sigma_1'(\chi), \\
0=&2p(\chi)p'(\chi)+2q(\chi)q'(\chi)+2p'(\chi)\sigma_1(\chi)+2p(\chi)\sigma_1'(\chi), \\
        a_0'=&2(p(\chi)p'(\chi)+q(\chi)q'(\chi))\sigma_1(\chi)+(p^2(\chi)+q^2(\chi))\sigma_1'(\chi).
    \end{aligned}
    \right.
\end{align*}
We evaluate at $\chi = \chi_*$, where $p(\chi_*) = 0$, $q(\chi_*) = \sqrt{a_1(\mu)}$, and $\sigma_1(\chi_*) = -a_2(\mu)$. 
Thus,
\begin{align*}
    \left\{
    \begin{aligned}
        0 &= 2p'(\chi_*) + \sigma_1'(\chi_*), \\
        0 &= 2\sqrt{a_1}\, q'(\chi_*) - 2p'(\chi_*)a_2, \\
        a_0' &= -2\sqrt{a_1}\, q'(\chi_*)a_2 + a_1\sigma_1'(\chi_*).
    \end{aligned}
    \right.
\end{align*}
Solving this system for $(p'(\chi_*),\, q'(\chi_*),\, \sigma_1'(\chi_*))^T$, we obtain
\begin{equation*}
    p'(\chi_*) = -\frac{2a_0'}{a_1(\mu) + a_2^2(\mu)} < 0,
\end{equation*}
where the inequality follows from the assumption that $\mu$ lies outside the interval where $a_1(\mu) < 0$. 
    
\end{proof}

\begin{remark}
The main difference between our proof and that in \cite{rodriguez2021understanding} is that our assumptions are based on the eigenvalues of $-\Delta$, rather than on the wave mode vector $\vec{k}$, which reduces to a scalar in one dimension. One could also verify the stability of the resulting periodic solutions following the same methodology as in that paper. Since the procedure does not directly depend on the spatial dimension, we expect analogous conclusions when $\chi$ is treated as a function of $\mu$, and therefore the corresponding proof is omitted. The primary purpose of this section is to illustrate how eigenvalues of the negative Laplacian can be employed to obtain alternative, dimension-independent results, using Hopf bifurcation as an illustrative example.
\end{remark}

The bifurcation value $\chi_*$ determines the onset of Hopf instabilities, where steady states lose stability and give rise to oscillatory behavior. In social-science terms, this suggests that when policing intensity crosses these thresholds, crime dynamics may enter into recurring cycles—periods of high and low crime concentrated in particular neighborhoods—rather than settling into a stable equilibrium. Such cyclical hotspot activity has been observed empirically and reflects the limitations of static policing strategies.

\section{Numerical Simulations and Pattern Formation}

In this section, we numerically investigate heterogeneous steady states in two dimensions, a natural setting for studying urban crime since it reflects how hotspots emerge and persist across neighborhoods.  Previous studies \cite{rodriguez2021understanding, yerlanov2025} employed linear stability analysis to identify key thresholds, denoted by $\chi^-$ and $\chi^+$ (see \cite{yerlanov2025} for precise definitions). The constant steady state is linearly stable whenever $\chi^- < \chi < \chi^+$.   Crossing these thresholds corresponds to intensities at which uniform crime–guardian equilibria destabilize, producing either stationary clustering (steady-state bifurcation) or cycling hotspots (Hopf).  From a social perspective, these bifurcation values mark thresholds of police responsiveness: when policing intensity remains between $\chi^-$ and $\chi^+$, uniform safety can be maintained, but crossing either boundary destabilizes the system, potentially leading to persistent or oscillatory crime concentrations.

Here, we focus on cases in which at most one of these inequalities is violated.  For reproducibility and further exploration, the simulation codes are publicly available: parameter exploration can be carried out with \url{https://github.com/MaYeatCo/urban_bifurcation}, while individual one- and two-dimensional simulations can be found in \url{https://github.com/MaYeatCo/urban_suppression}.

Since the number of parameters is large, we fix most of them and focus on how the diffusion rates affect the solutions. Specifically, we consider the case of hotspot policing, where police movement mimics that of crime agents, \textit{i.e.}, $\chi = 2$ and $D_u = D_\rho$. Our simulations are conducted in two dimensions on a square domain with side length $L$.  We perform a parameter sweep over a fixed range of $D_A$ and $D_\rho$ and compute the orderings of $\chi$, $\chi^-$, and $\chi^+$. The results are summarized in Figure \ref{fig:theory}. All possible orderings are represented by regions I–VI, and these regions correspond to different policing thresholds, where small changes in deployment intensity can qualitatively alter crime distributions.  For instance, the border between regions I and III occurs precisely when $a_2(\mu)a_1(\mu) = a_0(\mu)$ in Theorem \ref{3rdthm}.  Our objective is to determine whether periodic or near-periodic solutions emerge, thereby confirming the theoretical predictions.  After identifying the stability regions, we carry out numerical simulations. We select three parameter combinations within the interior of the regions (labeled A–C), as well as three parameter combinations located around the borders between region I and the others (labeled D–F). Here, we did not conduct simulations for regions IV–VI in Figure \ref{fig:theory}. These regions lie farther from Region I and correspond to cases where multiple Routh–Hurwitz conditions are violated, making the system’s behavior less predictable. As shown in \cite{rodriguez2021understanding}, solutions in these regimes can exhibit chaotic dynamics. Examples of such less predictable behavior (in the theoretical sense) will be presented in the following subsections.

\begin{figure}[h!]
    \centering
    \includegraphics[width=0.9\linewidth]{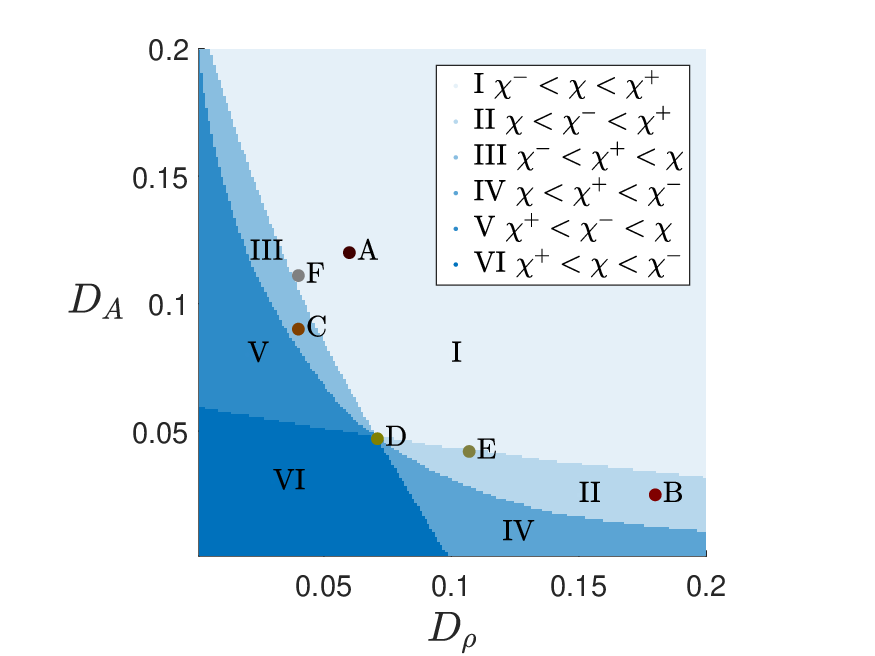}
    \caption{
Illustration of the effect of diffusion rates on the ordering of $\chi$, $\chi^-$, and $\chi^+$ obtained from linear stability analysis. Roman numerals indicate the corresponding inequality regimes, while letters denote parameter combinations selected for further simulations. Region I corresponds to the linearly stable regime. The remaining parameters are fixed at $\alpha = 1.0, \beta = 1.0, \lambda = 1.5, \chi = 2, L = \pi$. Each region represents a distinct ordering of $\chi, \chi^-$, and $\chi^+$.
   Note that this figure illustrates only the predictions from linear stability analysis. Nonlinear effects may significantly influence the solution behavior, particularly near the boundaries between regions. }
    \label{fig:theory}
\end{figure}

 We use the PDE toolbox in MATLAB (\cite{MATLAB, PdeToolbox}) to perform numerical simulations.  As simulations are done in two dimensions, it is not straightforward to depict the evolution of the solution. Thus, we choose to track changes in overall amplitude via a root mean square function. All solutions are initiated at the steady-state constant plus a small perturbation ($0.01$) using a cosine function. The time evolutions are shown in Figure \ref{fig:RMS}. We use the root mean square metric, defined as $\text{RMS}[f(x, t)]=\sqrt{|\Omega|^{-1}\int_\Omega f(\mathbf{x}, t)^2d\mathbf{x}}$. In addition, we plot the final frames of the simulation to show the patterns they exhibit, see Figure \ref{fig:tails}.
 
We now turn to Figures \ref{fig:RMS} and \ref{fig:tails} for a more detailed examination of the dynamics. At point A, corresponding to region I, perturbations decay monotonically, and the solution converges to the spatially uniform steady state, in full agreement with the linear stability predictions. At points B, D, and E, the solution amplitude increases and then saturates at a nonzero level, yielding spatially heterogeneous but temporally stationary states. These patterns may appear as hotspots, symmetric arrangements, or related structures.  At point C, located in region III, the long-time dynamics approach a time-periodic orbit. Near regime boundaries, as exemplified by point F, the system exhibits more intricate dynamics, with oscillations of varying amplitude that indicate the onset of mixed-type or higher-order bifurcations. Such transitional regimes call for more refined analysis, either through extended simulations or by applying weakly nonlinear techniques.  In summary, the numerical results confirm the existence of periodic solutions, consistent with the statement of Theorem \ref{3rdthm}.

\begin{figure}[H]
    \centering
    \begin{subfigure}{0.43\textwidth}
     \includegraphics[width=\textwidth]{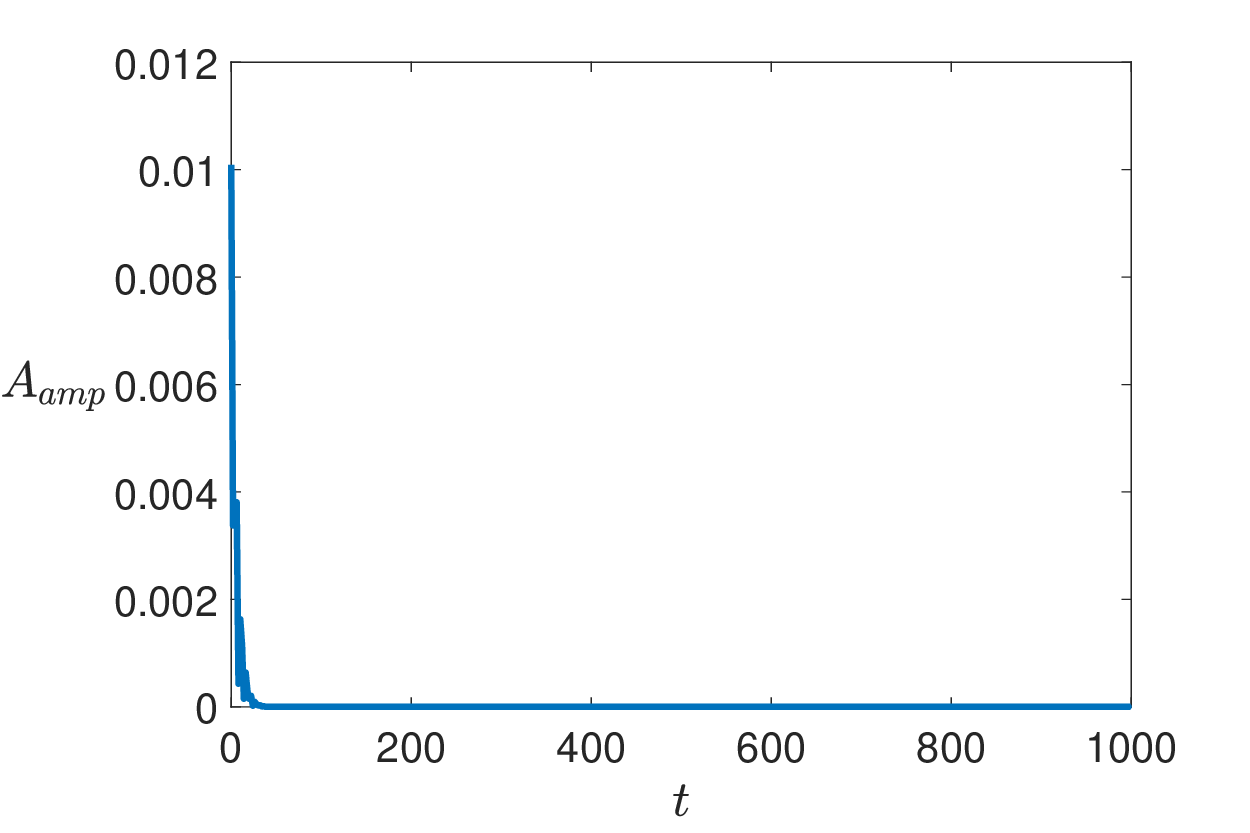}
     \caption{Combination A}\label{sfig:A_amp}
    \end{subfigure}
    ~
    \begin{subfigure}{0.43\textwidth}
    \includegraphics[width=\textwidth]{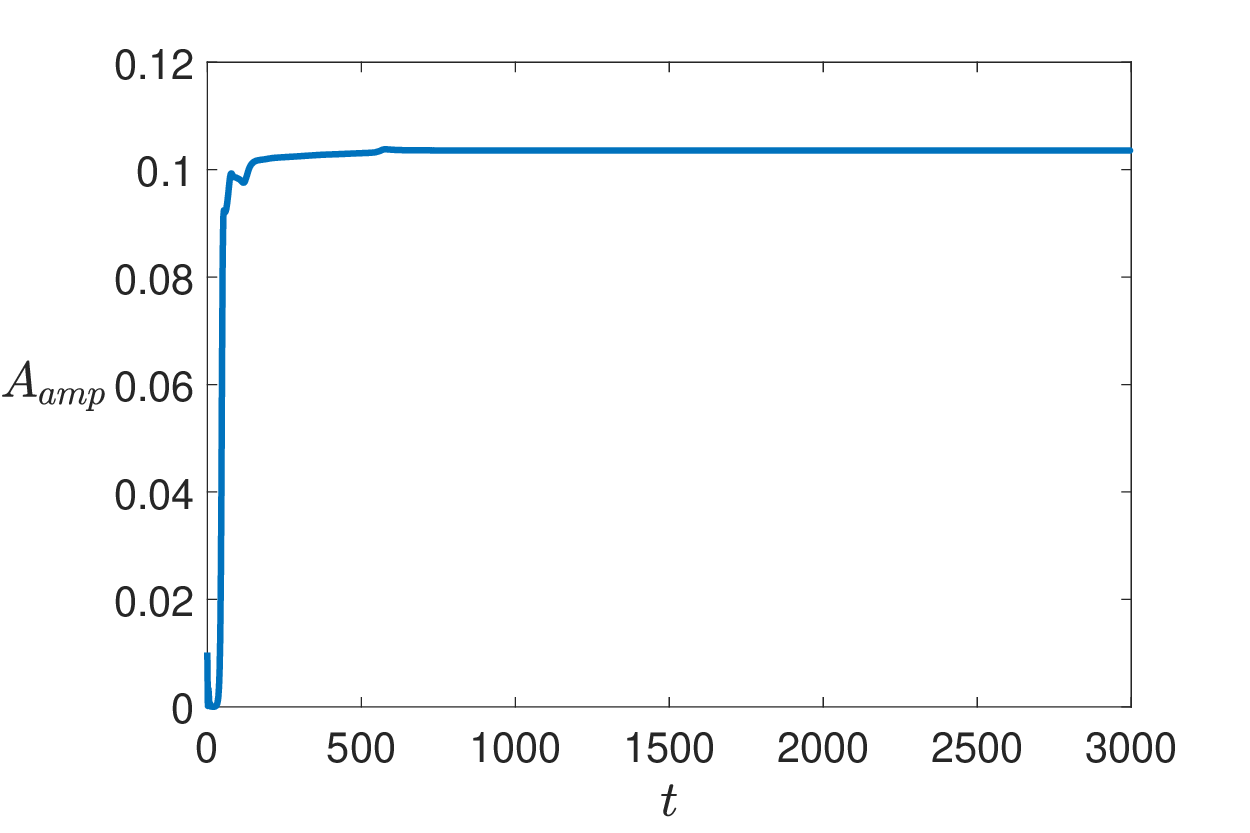}
    \caption{Combination B}\label{sfig:B_amp}
    \end{subfigure}
    \vskip\baselineskip
    \begin{subfigure}{0.43\textwidth}
     \includegraphics[width=\textwidth]{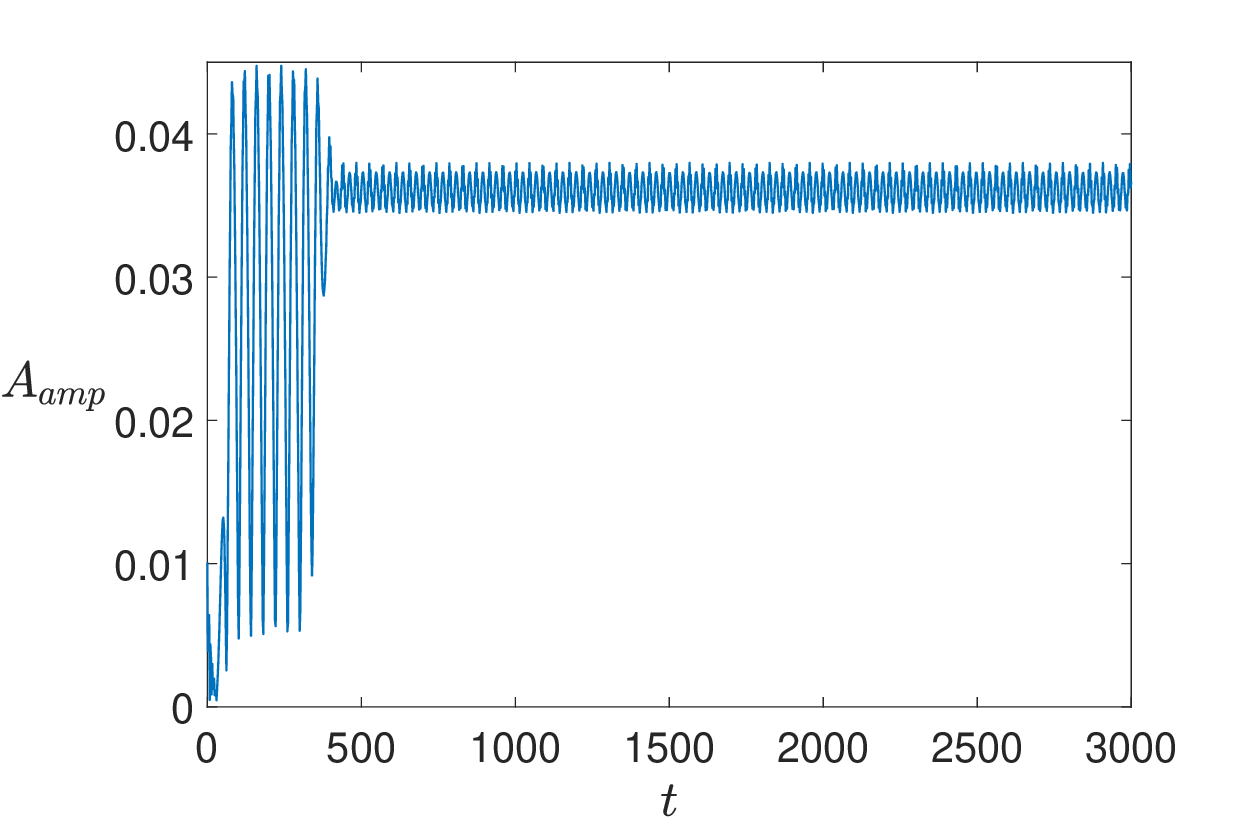}
     \caption{Combination C}\label{sfig:C_amp}
    \end{subfigure}
    ~
    \begin{subfigure}{0.43\textwidth}
    \includegraphics[width=\textwidth]{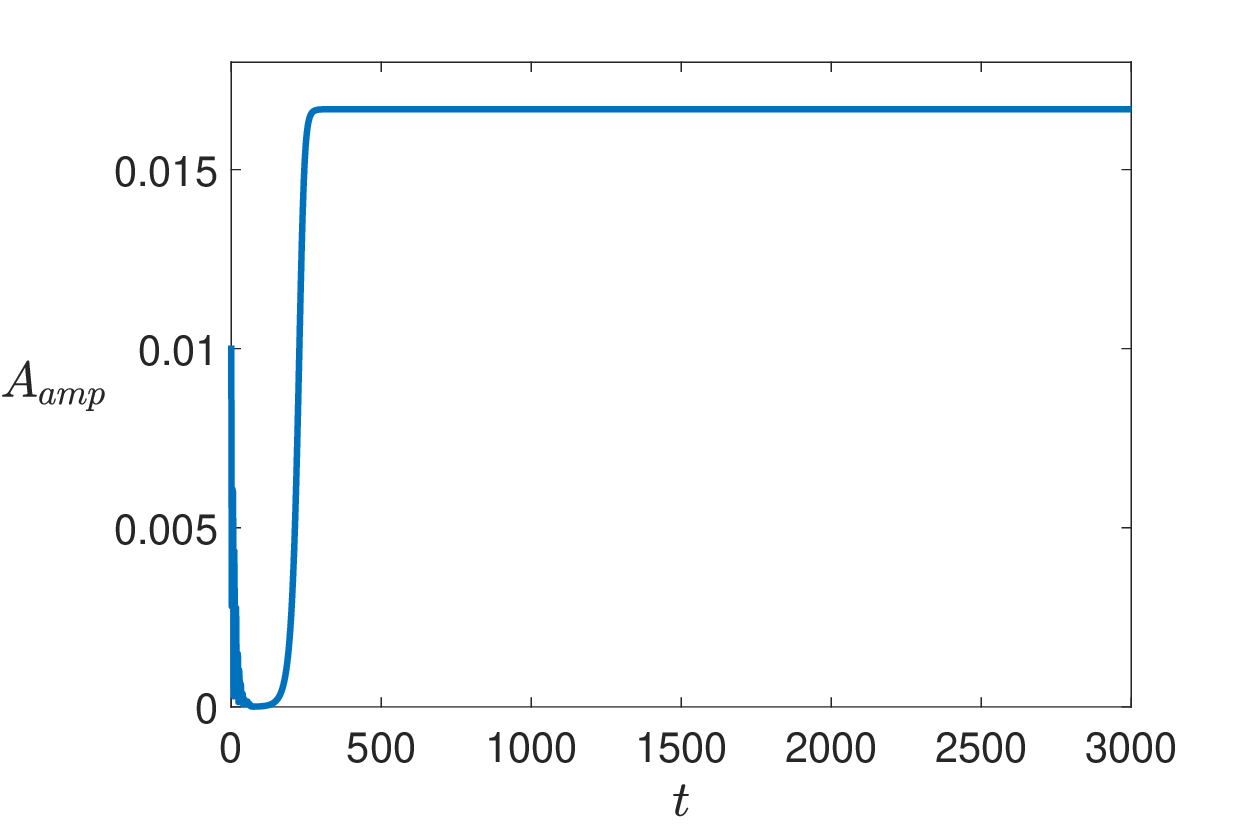}
    \caption{Combination D}\label{sfig:D_amp}
    \end{subfigure}
    \vskip\baselineskip
    \begin{subfigure}{0.43\textwidth}
     \includegraphics[width=\textwidth]{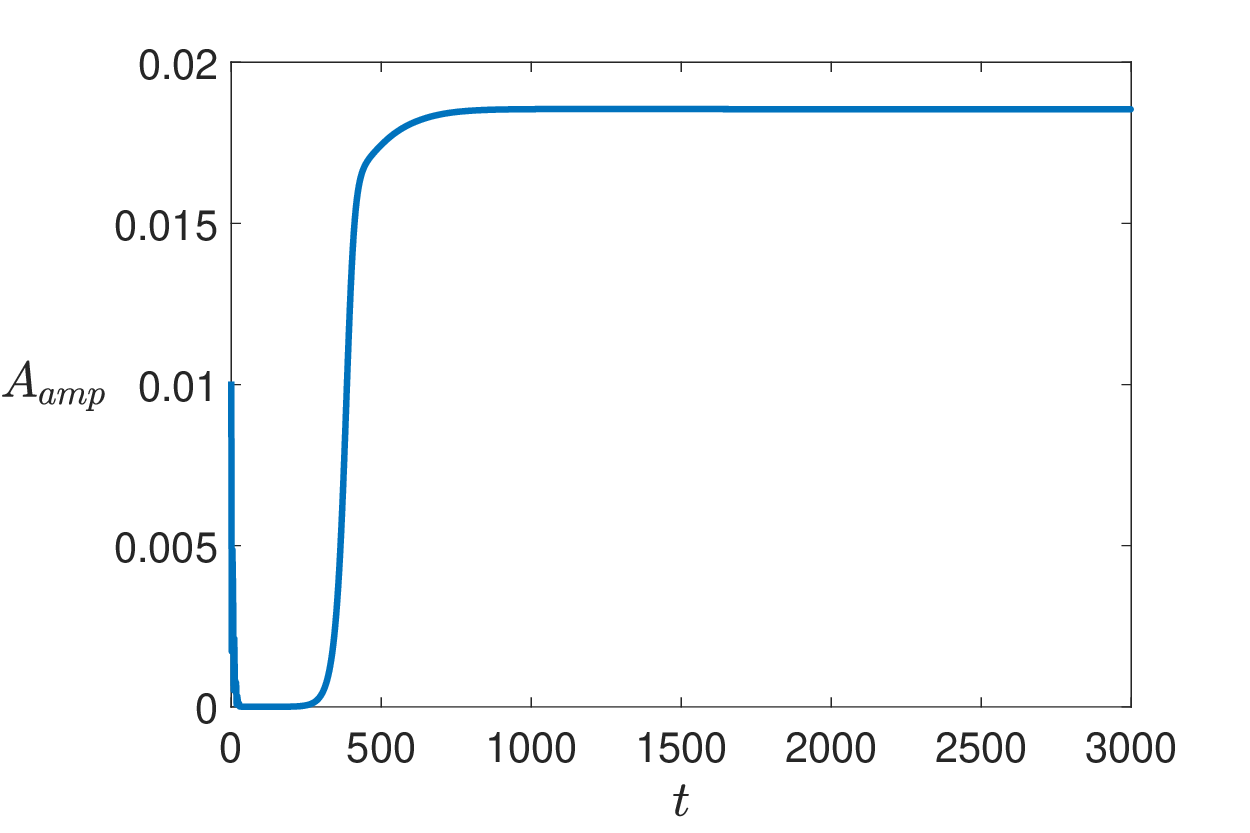}
     \caption{Combination E}\label{sfig:E_amp}
    \end{subfigure}
    ~
    \begin{subfigure}{0.43\textwidth}
    \includegraphics[width=\textwidth]{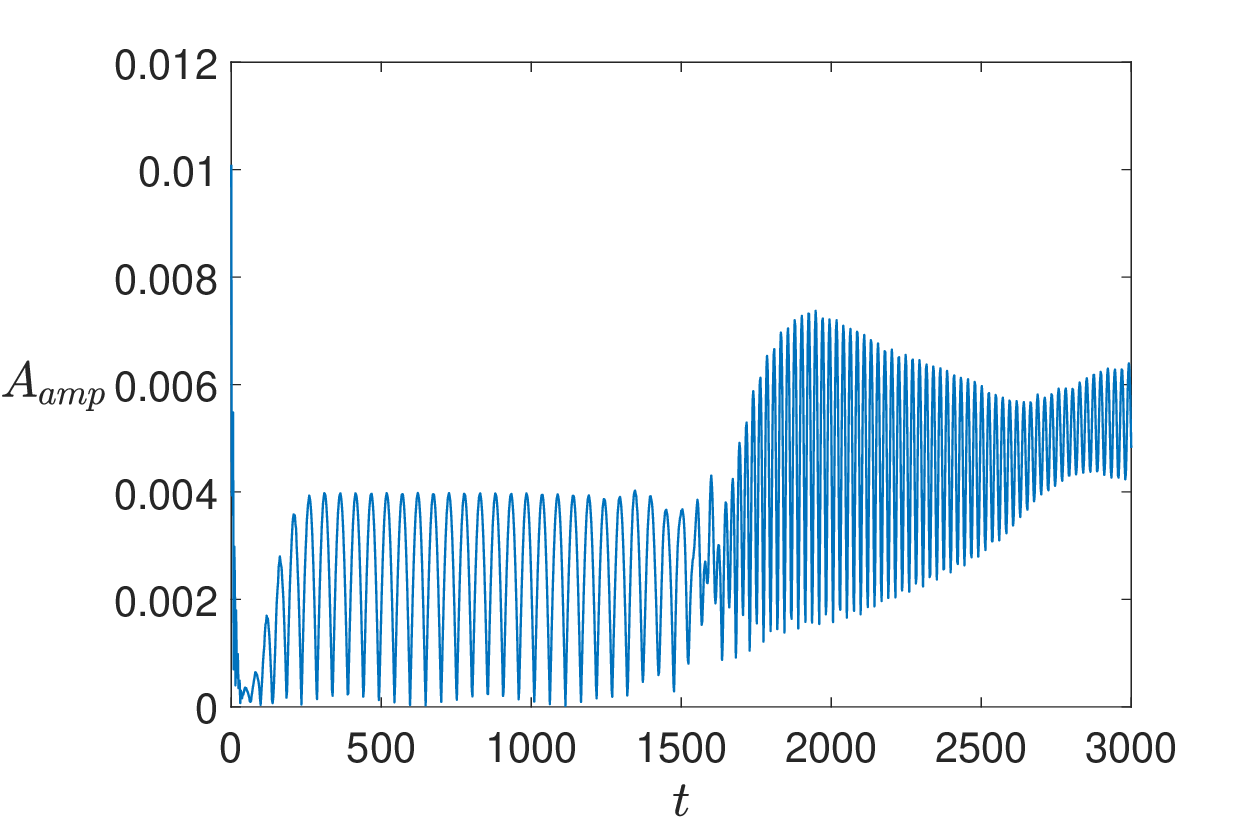}
    \caption{Combination F}\label{sfig:F_amp}
    \end{subfigure}
    \caption{
Examples of amplitude evolutions from Figure~\ref{fig:theory}, computed using $A_{amp}=\mathrm{RMS}[A(\mathbf{x},t)-\lambda]$, illustrate how diffusion rates shape system dynamics. 
In \textbf{(a)}, amplitudes decay monotonically, confirming the stability of the uniform steady state. In \textbf{(b)}, they saturate at a nonzero level, producing stationary hotspots. In \textbf{(c)}, periodic oscillations arise through a Hopf bifurcation, resembling recurrent “crime waves.” In \textbf{(d)}, irregular oscillations signal mixed or higher-order bifurcations and parameter sensitivity. Near regime borders, as in \textbf{(e)} and \textbf{(f)}, small perturbations can trigger clustering or quasi-periodic oscillations with irregular amplitudes, marking transitions to more complex or weakly chaotic behavior. Together, these results show how the system shifts from uniform stability to persistent hotspots, oscillatory patterns, and chaotic regimes as diffusion rates vary.
    }
    \label{fig:RMS}
\end{figure}

\begin{figure}[H]
    \centering
    \begin{subfigure}{0.4\textwidth}
     \includegraphics[width=\textwidth]{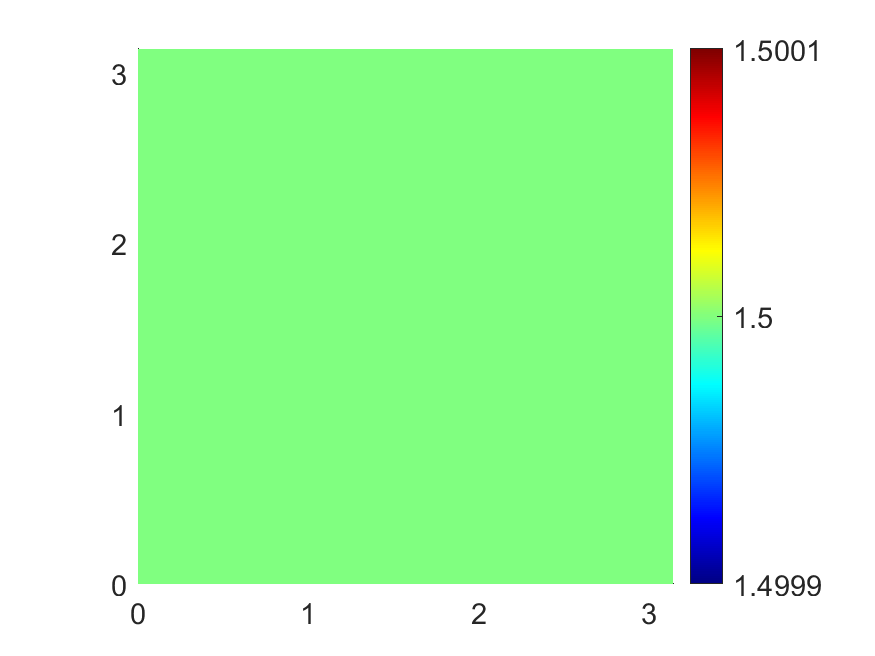}
     \caption{Combination A}\label{sfig:A_final}
    \end{subfigure}
    ~
    \begin{subfigure}{0.4\textwidth}
    \includegraphics[width=\textwidth]{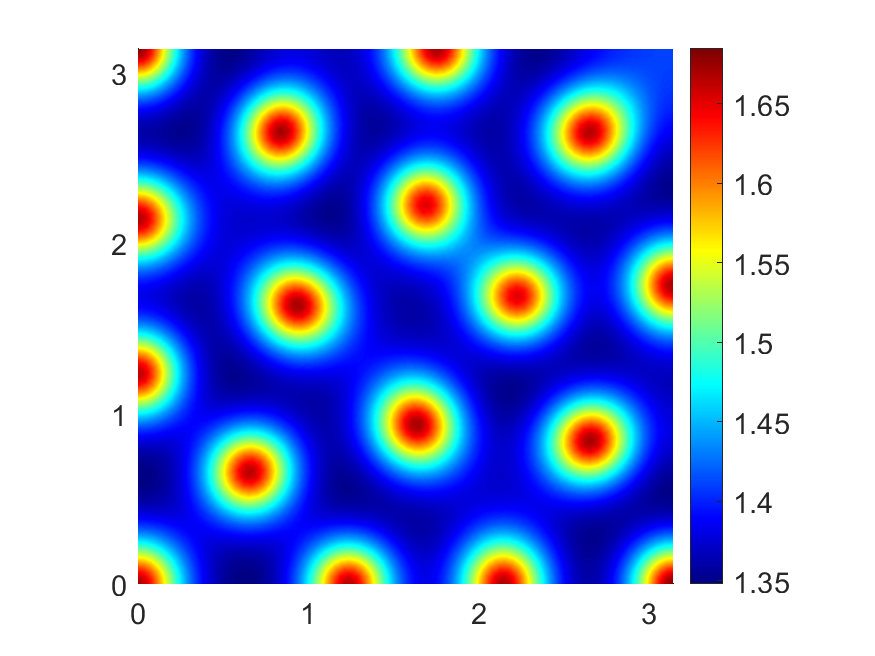}
    \caption{Combination B}\label{sfig:B_final}
    \end{subfigure}
    \vskip\baselineskip
    \begin{subfigure}{0.4\textwidth}
     \includegraphics[width=\textwidth]{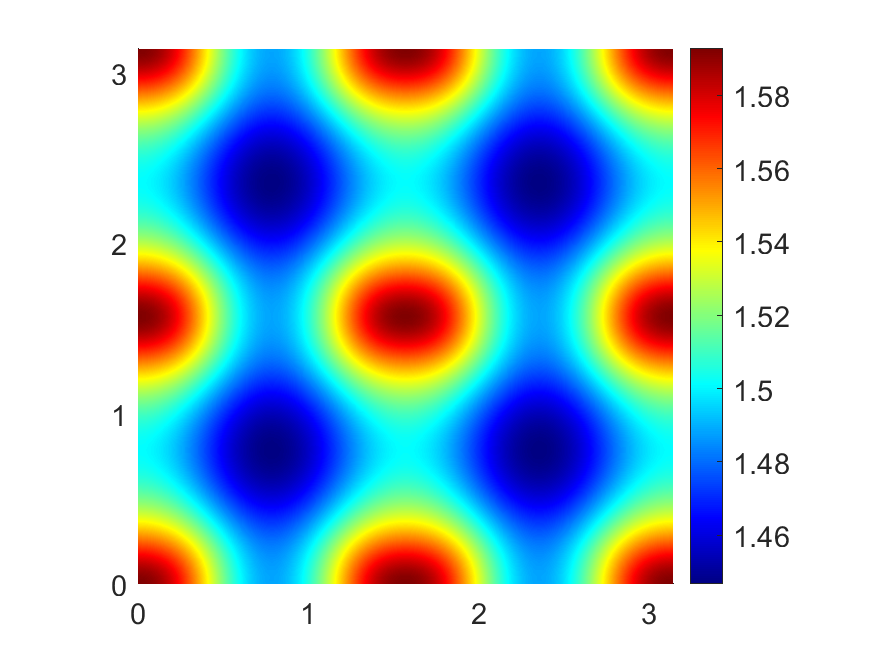}
     \caption{Combination C}
     \label{sfig:C_final}
    \end{subfigure}
    ~
    \begin{subfigure}{0.4\textwidth}
    \includegraphics[width=\textwidth]{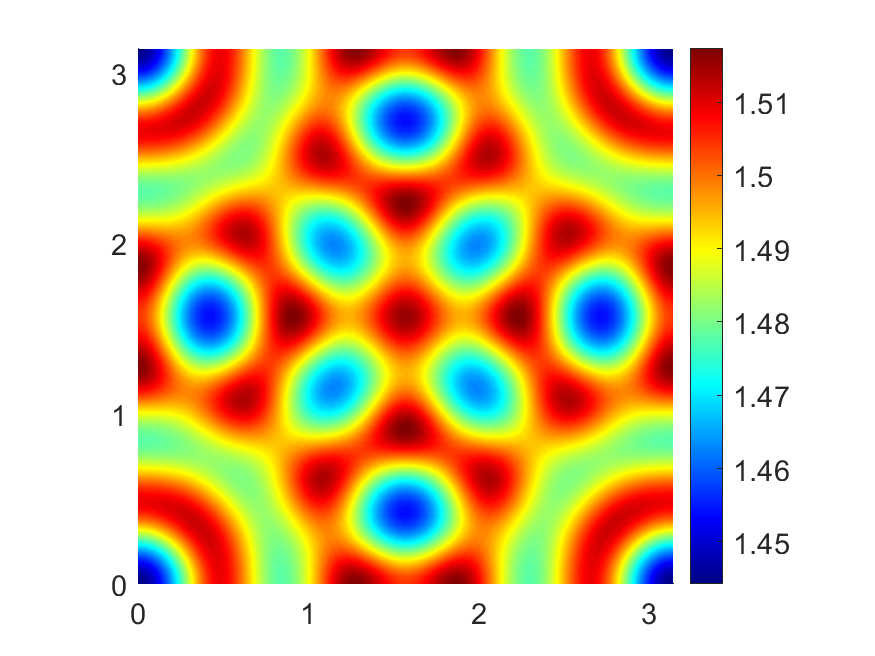}
    \caption{Combination D}\label{sfig:D_final}
    \end{subfigure}
\vskip\baselineskip
    \begin{subfigure}{0.4\textwidth}
     \includegraphics[width=\textwidth]{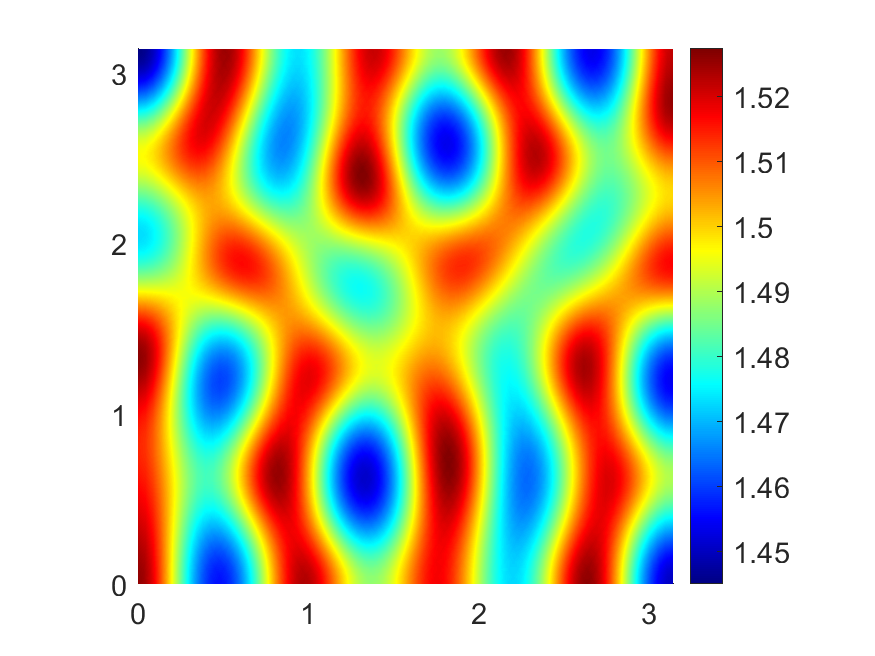}
     \caption{Combination E}\label{sfig:E_final}
    \end{subfigure}
    ~
    \begin{subfigure}{0.4\textwidth}
    \includegraphics[width=\textwidth]{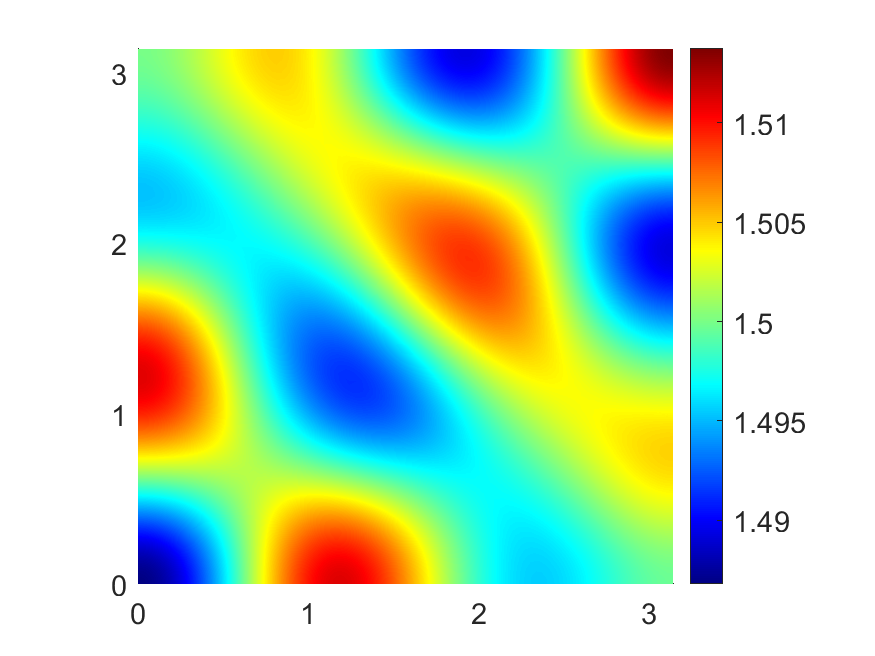}
    \caption{Combination F}\label{sfig:F_final}
    \end{subfigure}
    \caption{
Examples of patterns corresponding to the points in Figure~\ref{fig:theory}. The diffusion rates are: \textbf{(a)} Region I, $D_A=0.12, D_\rho=0.06$; \textbf{(b)} Region II, $D_A=0.025, D_\rho=0.18$; \textbf{(c)} Region III, $D_A=0.09, D_\rho=0.04$; \textbf{(d)} intersection of all regions, $D_A=0.047, D_\rho=0.071$; \textbf{(e)} near the border between Regions I and II, $D_A=0.042, D_\rho=0.107$; \textbf{(f)} near the border between Regions I and III, $D_A=0.111, D_\rho=0.04$. The remaining parameters are fixed at $\alpha=0.5, \beta=1.0, \lambda=0.75, \chi=2, L=\pi$.  The snapshots reveal distinct regimes: uniform states \textbf{(a)}, stationary hotspots \textbf{(b)}, periodic oscillations \textbf{(c)}, irregular mixed patterns at parameter intersections \textbf{(d)}, and near-threshold behavior with localized clustering or quasi-periodic fluctuations \textbf{(e)}, \textbf{(f)}. Together, they illustrate how tuning diffusion rates drives transitions from uniform stability to heterogeneous clustering, oscillations, and complex dynamics.    
}
    \label{fig:tails}
\end{figure}

\subsection{Time-periodic patterns}

In two-dimensional simulations, the solution may require a considerable time to converge, as can be observed in Figure \ref{sfig:F_amp}. To investigate periodic behavior more efficiently, we therefore turn to a one-dimensional setting, where simulations are computationally less demanding. This reduction not only decreases the run-time but also allows the use of a smaller time step $\Delta t$, improving the numerical accuracy. While simplified, the one-dimensional results remain informative, as they capture the essential dynamics and can be extended qualitatively to the two-dimensional case, particularly concerning the existence of distinct classes of solution behavior.

To demonstrate the comparison between the two domains, we simulate system \eqref{eq:E} over the interval under the same combination as in Figure \ref{sfig:C_amp}. The results are illustrated in Figure \ref{fig:1dverofC}. Figure \ref{fig:1dampofC} is the one-dimensional counterpart of Figure \ref{sfig:C_amp}. Although there are differences, the overall qualitative behavior is the same between the two.  Moreover, the simulations over an interval take a shorter time to stabilize (200 in 1-D versus 500 in 2-D), and we see a clearer periodic solution, {\it i.e.} we are more confident in characterizing this behavior as periodic. This trend is further supported in Figure \ref{fig:1dphaseofC}, where we observe a limit cycle. The period seems to be $\approx 240$ time units. Further time-series analysis can be performed, which can be helpful when examining multiple simulations, such as a parameter sweep setting in Figure \ref{fig:theory}, but this is beyond the scope of this paper. Instead, in the following section, we utilize the speed and reliability of one-dimensional simulations to shed light on another type of solution: chaos. 

\begin{figure}[H]
    \centering
    \begin{subfigure}{0.54\textwidth}
    \includegraphics[width=\textwidth]{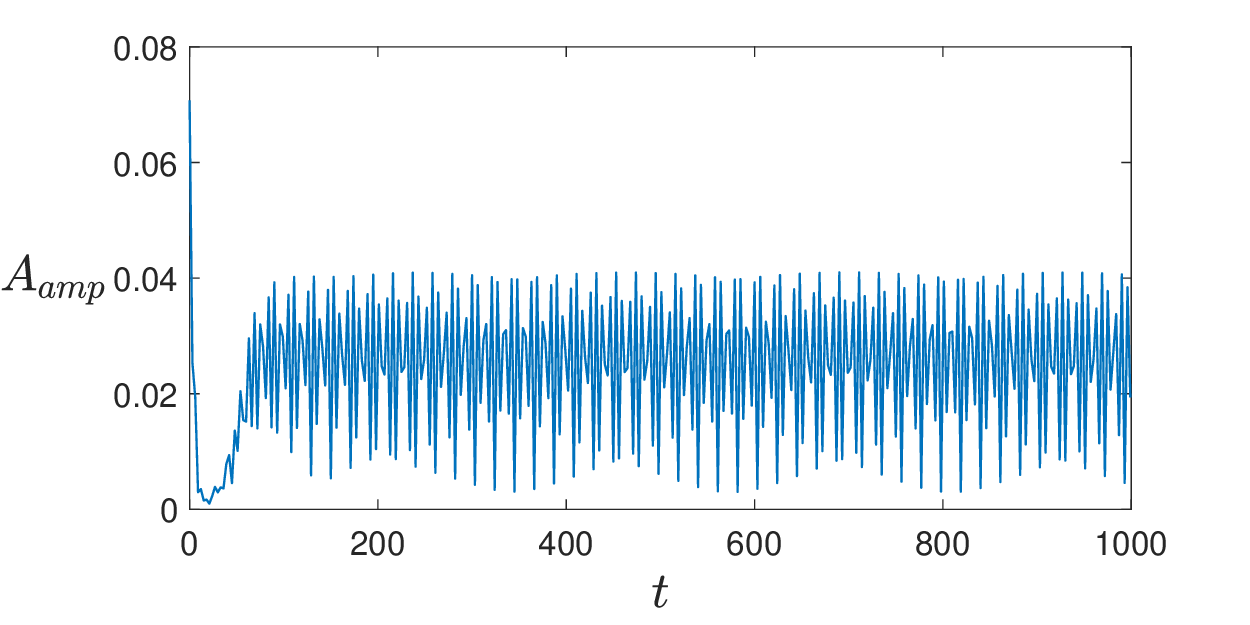}
    \caption{Amplitude evolution of $A$}\label{fig:1dampofC}
    \end{subfigure}
    ~
    \begin{subfigure}{0.36\textwidth}
        \includegraphics[width=\textwidth]{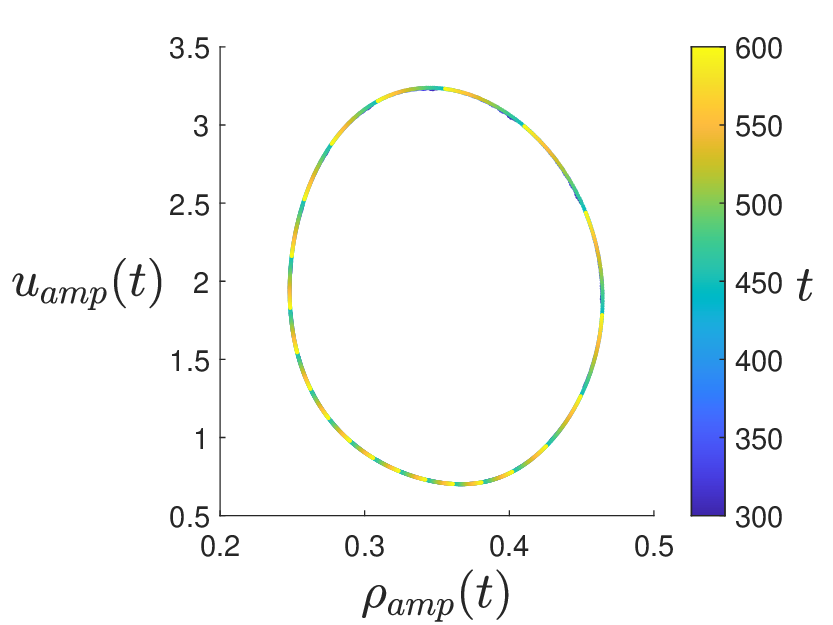}
        \caption{Phase-portrait of amplitudes $\rho$-$u$}\label{fig:1dphaseofC}
    \end{subfigure}
    \caption{
Representation of a one-dimensional simulation of \eqref{eq:E} with the same parameters as in \ref{sfig:C_amp}. 
\textbf{(a)} Emergence and stabilization of periodic behavior in the amplitude of the $A$ component, $A_{amp}=\mathrm{RMS}[A(x,t)-\lambda]$. Data are sampled every 3 time units to provide a clearer visualization of the oscillations. The figure shows that, after an initial transient, the solution does not converge to a steady state but instead settles into sustained periodic dynamics, reflecting the onset of a Hopf bifurcation. 
\textbf{(b)} Limit cycle in the $\rho_{amp}$ versus $u_{amp}$=($\mathrm{RMS}[\rho(x,t)-\rbar]$ versus $\mathrm{RMS}[u(x,t)-\ubar]$) phase portrait for $t \in [300,600]$, with trajectories represented in color. The closed orbit indicates that the dynamics of the offender density ($\rho$) and the guardian density ($u$) are locked into a recurring cycle rather than stabilizing to an equilibrium.  These results highlight the transition from stationary to oscillatory crime–policing patterns. In dynamical systems terms, the system crosses a bifurcation threshold where uniform hotspots lose stability and periodic oscillations emerge. From a criminological perspective, this suggests that crime and enforcement densities may fluctuate in sustained cycles—so hotspots not only persist but also oscillate in intensity and location over time, echoing empirical observations of recurring ``crime waves.”}
    \label{fig:1dverofC}
\end{figure}

\subsection{Exploring chaos}

A natural question is what occurs in regions IV–VI of Figure \ref{fig:theory}. As shown in \cite{rodriguez2021understanding}, under certain parameter regimes—for instance, when diffusion rates are sufficiently small—the system can exhibit chaotic behavior. To provide further evidence of this phenomenon, we present the bifurcation diagram in Figure \ref{fig:bifurd}.  The diagram reveals that as diffusion decreases, new dynamical behaviors emerge. The first periodic solution appears around $D_\rho = 0.06$. At $D_\rho = 0.03$, successive period-doublings occur, eventually leading to the onset of chaos. We explore these regimes through representative solutions in Figures \ref{fig:chaos_ex} and \ref{fig:period_ex}. In Figure \ref{sfig:chaosA}, the time series of $A_m(t)$ displays unpredictable fluctuations, and the system quickly settles into this chaotic regime. By contrast, Figure \ref{sfig:periodA} shows a clear periodic orbit, though it takes longer for the dynamics to stabilize, and the oscillation frequency is relatively high ($\gg 1$).

The corresponding phase portraits for $\rho$ and $u$ further illustrate the difference: in the chaotic case (Figure \ref{sfig:chaosrhoU}), trajectories approach a strange attractor, whereas in the periodic case (Figure \ref{sfig:periodrhoU}), a well-defined closed orbit is observed. Chaotic solutions also exhibit larger deviations from the steady state. For example, although $\ubar = 1$ in both cases, in the periodic regime $u$ oscillates between approximately $0.65$ and $1.2$, while in the chaotic regime it can range from near zero to as high as nine.  Many additional phenomena could be uncovered with a more detailed analysis. However, it is worth noting that the theory of chaos in infinite-dimensional PDE systems is less developed than in finite dimensions, making rigorous characterization challenging. Nevertheless, the present model appears to provide a promising and mathematically rich setting for future research on spatiotemporal chaos in criminology-inspired systems.

\begin{figure}[H]
    \centering
    \includegraphics[width=0.8\textwidth]{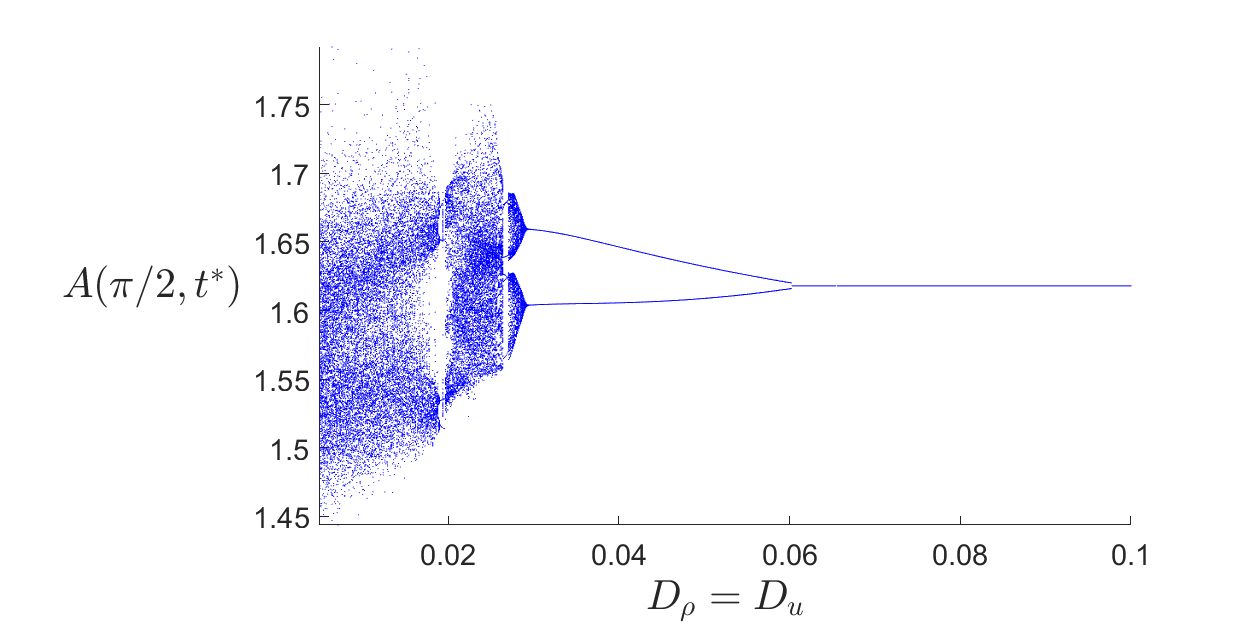}
    \caption{
Bifurcation diagram of the midpoint value $A(\pi/2,t^*)$ in the one-dimensional simulation of \eqref{eq:E} in $(0,\pi)$ with initial perturbation of the form $0.05\cos(x)$. The $y$-axis records all extreme values of $A(\pi/2,t)$ with $t \in [800,1000]$, after transients have decayed. $t^*$ denotes a maximizer or a minimizer. Simulations are run up to $T_{\max}=1000$, treating $t \in [0,800]$ as transient. Parameters are fixed at $D_A=0.1$, $\alpha=\beta=1$, $\lambda=(1+\sqrt{5})/2$, and $\chi=2$, with varying $D_\rho$ to reveal the bifurcation structure.  The diagram illustrates how decreasing $D_\rho$ drives qualitative changes in the long-term behavior of the system. For relatively large diffusion rates, solutions converge to spatially uniform steady states, consistent with linear stability predictions. As $D_\rho$ decreases, periodic solutions emerge, followed by successive period-doublings, and eventually chaotic dynamics. This transition reflects the loss of stability of homogeneous states and the onset of complex spatiotemporal behavior. In particular, the appearance of chaos indicates that small diffusion rates can amplify nonlinear feedback, leading to unpredictable but persistent fluctuations in the crime–policing system. 
    } 
    \label{fig:bifurd}
\end{figure}

\begin{figure}[H]
    \centering
    \begin{subfigure}{0.54\textwidth}
     \includegraphics[width=\textwidth]{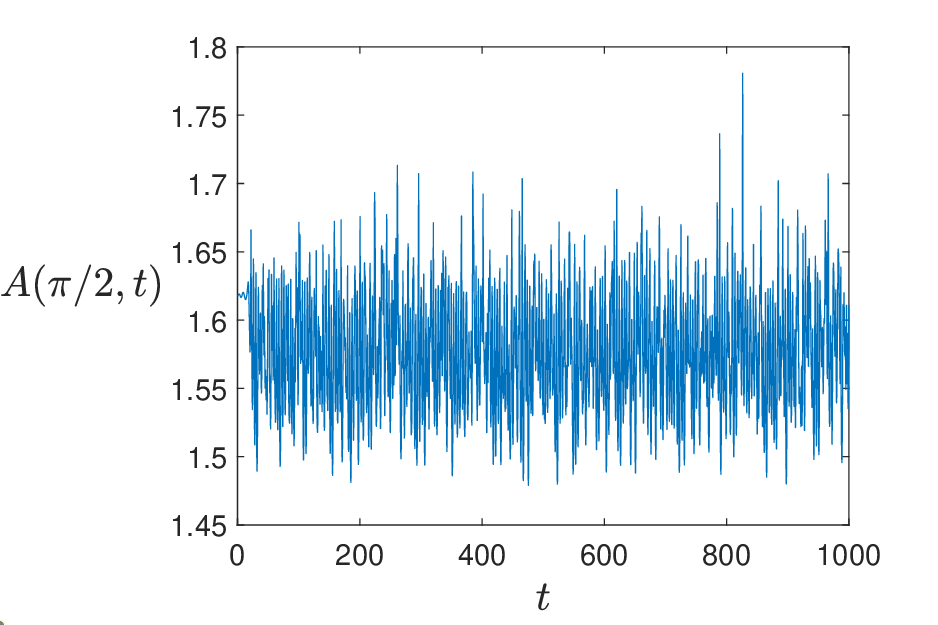}
     \caption{Time evolution of the midpoint of $A$}\label{sfig:chaosA}
    \end{subfigure}
    ~
    \begin{subfigure}{0.36\textwidth}
    \includegraphics[width=\textwidth]{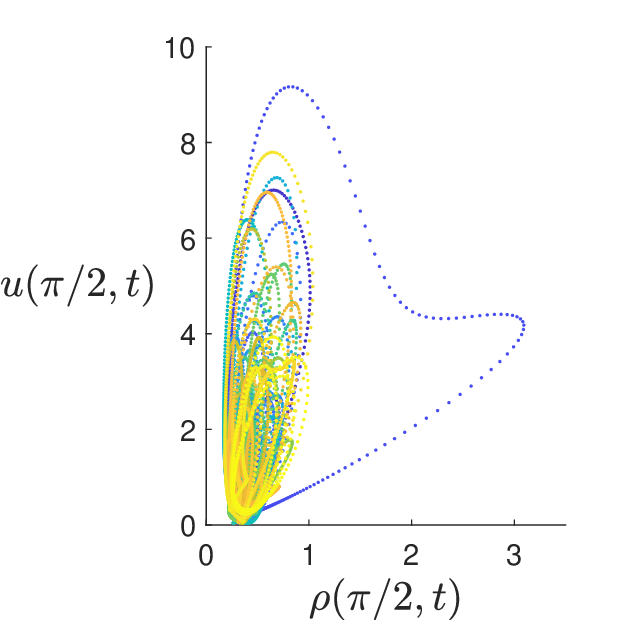}
    \caption{Phase-portrait of midpoints of $\rho-u$}\label{sfig:chaosrhoU}
    \end{subfigure}
    \caption{Representations of the chaotic solution to \eqref{eq:E} with the same parameters as in Figure \ref{fig:bifurd}, except that $D_\rho = D_u = 0.01$. As before, we investigate dynamics at the spatial midpoint, $x = \pi/2$.  \textbf{(a)} Time evolution of $A(\pi/2,t)$.
 Instead of converging to a steady state or periodic orbit, the trajectory exhibits irregular oscillations with no discernible repeating pattern, characteristic of chaos.  \textbf{(b)} Phase portrait of the non-transient trajectories of $\rho(\pi/2,t)$ and $u(\pi/2,t)$ for $t \in [800,1000]$, with colors indicating temporal progression (dark blue at $t = 800$, yellow at $t = 1000$).  The trajectory does not close into a limit cycle but instead wanders in a complex set consistent with a strange attractor. Together, these figures again illustrate that very small diffusion rates destabilize uniform or periodic states and produce chaotic dynamics. From a criminological perspective, this corresponds to volatile crime–policing interactions, where hotspots neither stabilize nor repeat regularly but instead fluctuate unpredictably over time.}
    \label{fig:chaos_ex}
\end{figure}

\begin{figure}[H]
    \centering
    \begin{subfigure}{0.54\textwidth}
\includegraphics[width=\textwidth]{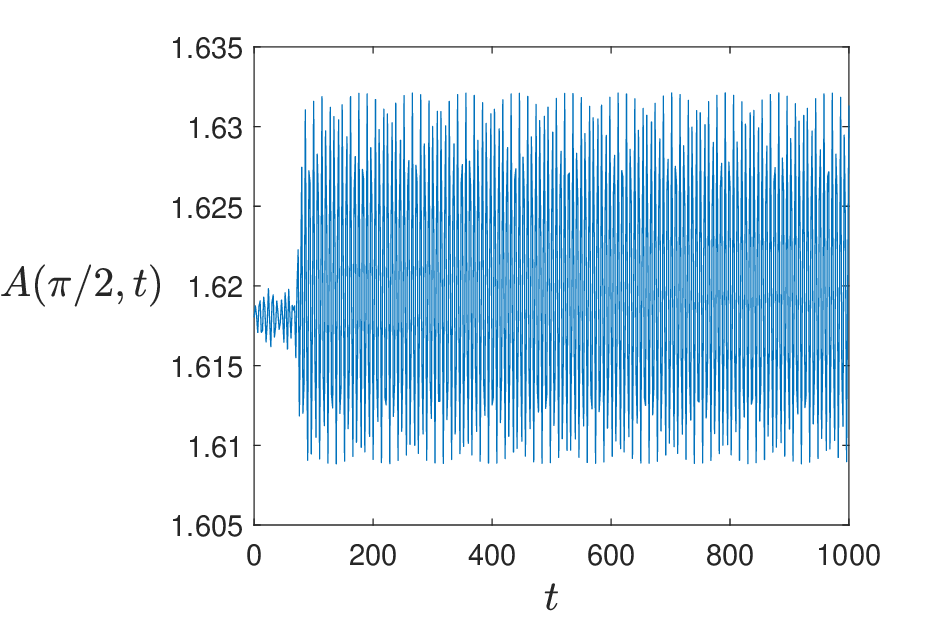}
     \caption{Time evolution of the midpoint of $A_m$}\label{sfig:periodA}
    \end{subfigure}
    ~
    \begin{subfigure}{0.36\textwidth}
\includegraphics[width=\textwidth]{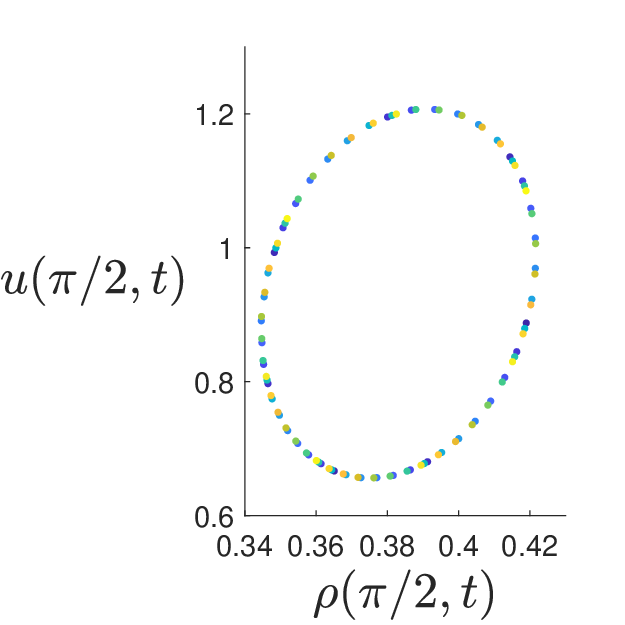}
    \caption{Phase-portrait of midpoints $\rho-u$}\label{sfig:periodrhoU}
    \end{subfigure}
    \caption{Representations of a periodic solution to \eqref{eq:E}. The parameters are the same as in Figure \ref{fig:bifurd}, but additionally $D_\rho=D_u=0.05$. \textbf{(a)} The full time of evolution of $A(\pi/2,t)$ is given. \textbf{(b)} The non-transient $\rho(\pi/2,t)$ and $u(\pi/2,t)$ are plotted against each other. The colors represent the time: dark blue -- $t=800$ and yellow -- $t=1000$.  When compared to Figure \ref{fig:chaos_ex}, we only slightly increased diffusion rates, yet now the trajectory is an orbit after passing the aperiodic transient state. Similarly to Figure \ref{fig:1dverofC}, we observe "crime waves". Moreover, the overall amplitudes are smaller. This figure, along with \ref{fig:chaos_ex}, demonstrates how changes in the agents' movement affect the solution.}
    \label{fig:period_ex}
\end{figure}

\subsection{Further exploration of parameter space}

In this section, we complement our theoretical findings by examining how parameter choices affect the behavior of solutions. Specifically, we focus on the movement dynamics of real agents—namely, how variations in the diffusion and advection rates of $\rho$ and $u$ influence the system. All other parameters are held fixed. In this way, our analysis can be interpreted as studying the outcomes of strategic decisions from one side: given the random movement of criminal agents, what spatial patterns or effects might arise under different police deployment strategies? To address this, we use an interval domain for phase-plane simulations, which enables us to employ sufficiently small time steps to capture the dynamics at lower diffusion values. Recall that large diffusion rates lead to a linearly stable regime, a case that has already been analyzed.

Similar to Figure \ref{fig:theory}, we examine the influence of diffusion rates, but with $D_A$ held fixed. In particular, we explore combinations of relatively small values of $D_\rho$ and $D_u$, with the results shown in Figure \ref{fig:dudw}. As expected, smaller diffusion rates lead to chaotic behavior. As diffusion rates increase, solutions tend to become spatially heterogeneous or exhibit periodic patterns. Only when both diffusion rates are sufficiently large ($D_\rho > 0.1$ and $D_u > 0.05$) does the system converge to constant solutions.  This relationship can be interpreted as crime-stabilizing strategies: police agents may increase their diffusion to drive the system toward a constant-in-time state. Moreover, we observe that temporally varying solutions typically attain significantly higher and lower amplitudes compared to the constant steady-state values, with smaller diffusion rates producing larger deviations. This observation suggests a potential crime-suppression strategy: wait until the density of criminal agents reaches a low point and then intervene to stabilize the system at that level (see \cite{yerlanov2025} for additional suppression strategies).

Next, we consider a scenario in which the diffusion rate of crime agents is fixed while the police advection rate is varied. This provides an alternative perspective for exploring potential strategies and policy outcomes. The corresponding results are shown in Figure \ref{fig:dwchi}. Recall that $\chi=2$ represents hotspot policing, where the police advection rate matches that of crime agents (second row). Increasing $\chi$ concentrates police more strongly around hotspots (first row), whereas flipping the sign corresponds to off-hotspot policing, where police avoid areas of high attractiveness (last row). We also examine the case with no advection (third row).  We observe that increasing $\chi$ while keeping $D_u$ small results in chaotic or periodic behavior, which can be attributed to the fact that $\chi > \chi^+$. By contrast, increasing $D_u$ appears to stabilize the system, at least in a temporal sense.

\begin{figure}[H]
    \centering
    \includegraphics[width=0.8\linewidth]{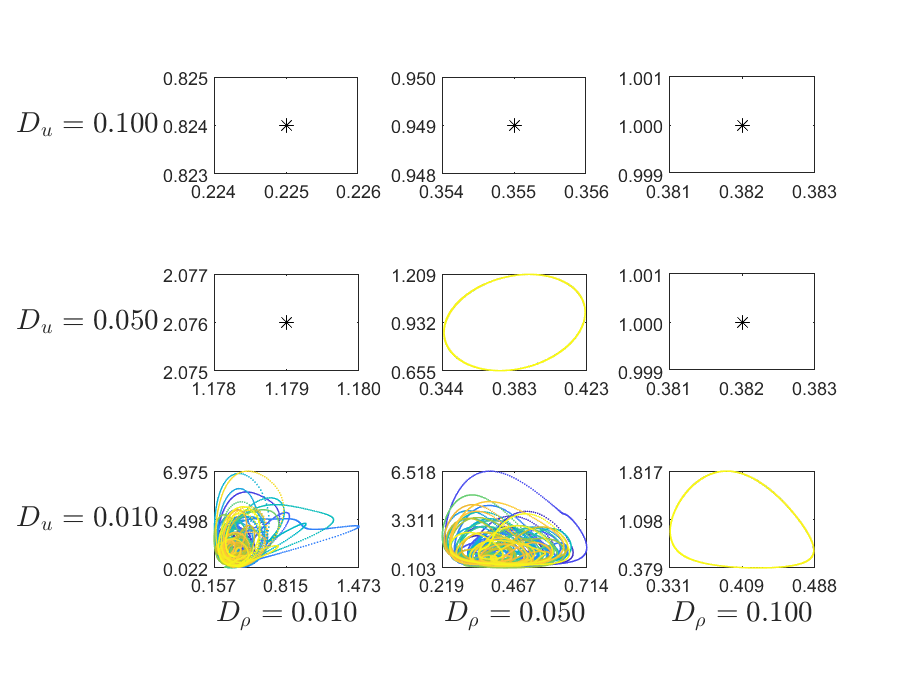}
    \caption{Solutions to \eqref{eq:E} for various $D_\rho$ (columns) and $D_u$ (rows). 
    The non-transient $\rho(\pi/2,t)$ and $u(\pi/2,t)$ are plotted against each other. The remaining parameters are: $D_A=0.1$, $\alpha=\beta=1$, $\lambda=(1+\sqrt{5})/2$, $\chi=2$. The colors represent the time: dark blue -- $t=800$ and yellow -- $t=1000$. The star represents the constant-in-time solutions. Whenever $\rho(\pi/2,t)=0.382$ and $u(\pi/2,t)=1$, this indicates constant-in-space solutions as well. We observe that the reduction in the diffusion, which is a proxy for random movement here, leads to less predictable behavior, destabilizing the crime landscape.}
    \label{fig:dudw}
\end{figure}

However, as shown by the linear stability analysis in \cite{yerlanov2025}, no value of $D_u$ with fixed $\chi = 4$ can shift the system into the linearly stable regime; in other words, there does not exist $D_u > 0$ such that $\chi^- < \chi < \chi^+$ for the parameters given in the caption of Figure \ref{fig:dwchi}.  In this setting, high police concentration in highly attractive regions prevents the system from reaching spatial uniformity, regardless of adjustments to patrolling (\textit{i.e.}, random movement).  In particular, when police movement has no advection component, the system decouples: $u$ converges to a constant steady state, while the remaining scalar fields develop spatially heterogeneous solutions. In this case, the model effectively reduces to a two-component system, which has already been extensively studied in the works referenced in the introduction.

In our simulations, off-hotspot policing is associated with higher crime density and lower police presence, whereas the opposite trend is observed under hotspot policing. For example, in the bottom-left corner of Figure \ref{fig:dwchi}, the crime density reaches $\approx 28.3$, about 73 times higher than the constant steady-state value. While such disparities may resemble the heterogeneous distribution of crime in urban areas (\textit{cf}, \cite{thacher2011distribution}), we emphasize that these findings are model-based and hypothetical. Further empirical validation would be required before drawing any policy implications.

\begin{figure}[H]
    \centering
    \includegraphics[width=0.8\linewidth]{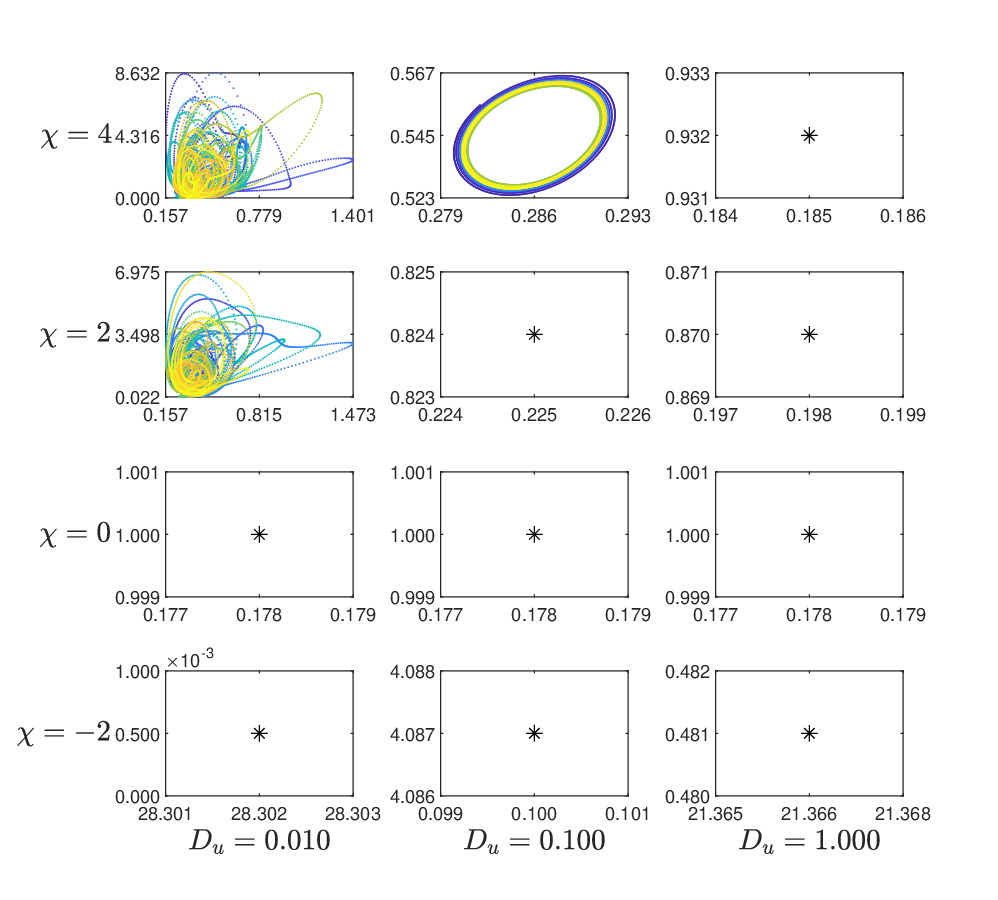}
    \caption{Solutions to \eqref{eq:E} for various $D_u$ (columns) and $\chi$ (rows). The non-transient $\rho(\pi/2, t)$ and $u(\pi/2,t)$ are plotted against each other. The remaining parameters are: $D_A=0.1$, $D_\rho=0.01$, $\alpha=\beta=1$, $\lambda=(1+\sqrt{5})/2$.  We represent two snapshots using dark blue ($t=800$) and yellow ($t=1000$). The star represents the constant-in-time solutions. Here, we note that not only diffusion rates, but other parameters as well, may affect the crime landscape. We observe that the magnitude and sign of $\chi$ can significantly impact the local density of both agents, leading to values far from the constant steady state. }
    \label{fig:dwchi}
\end{figure}

\section{Conclusion and Brief Discussions}

We have contributed to the understanding of the pattern formation of a model of urban crime by demonstrating the existence and stability of a nonhomogeneous solution bifurcating from the constant steady state. The bifurcation parameter, $\chi$, does not have any restrictions in terms of its sign and hence a bifurcating point is always defined. The corresponding spatially nonconstant solutions are almost always present, but they are not always stable. Our approach relies on connecting the central system with the negative Laplacian eigenvalue problem, introduced by Cantrell \textit{et al} \cite{cantrell2012global}.
This is a powerful tool, as it does not involve the wavemode vector and thus does not explicitly depend on the dimension of the problem. We employed this method to investigate the Hopf bifurcation and obtained a similar result to that derived for the one-dimensional case \cite{rodriguez2021understanding}. One could further extend and show the stability of periodic solutions, but we expect the results to be similar to those in \cite{rodriguez2021understanding}. In addition, we run numerical simulations of system \eqref{eq:E}, which confirm the theoretical results by demonstrating the existence of orbits. Even when solutions are not periodic, they still may exhibit interesting patterns, whose geometry can be studied as a separate future project. From a theoretical perspective, a worthwhile direction is to analyze chaos, which is only evidenced numerically in this paper, more rigorously. From a practical standpoint, a deeper understanding of various parameter combinations and corresponding solution behaviors needs to be studied as part of a potential comprehensive research design (\textit{e.g.}, how police should move given the current state of the system).  

As the work of Cantrell and collaborators \cite{cantrell2012global} complements the work by Short {\it et al.} \cite{short2010nonlinear}, this work complements our previous finding in \cite{yerlanov2025}, by providing the relevant Theorems. Although the approaches are similar, we note a few differences between this paper and that of Cantrell {\it et al.} \cite{cantrell2012global}. We have fewer restrictions on the parameter space for the first Theorem, as the control parameter does not have to be positive. On the other hand, we needed to impose more conditions to guarantee stability, as it is ``easier" for the relevant characteristic equation to have positive real parts. Cantrell {\it et al.} stated that the relevant eigenvalues are those that are sufficiently small. This does not appear to be the case for the extended model. We hypothesize that under suitable conditions, any eigenvalue of the negative Laplacian would be relevant. It would be of interest to find the equivalent but simpler assumptions in the Theorems. The apparent difference lies in the addition of the third equation, and hence, the increase in dimension from a linear algebra perspective. This comes with challenges, such as dealing with an already non-trivial kernel, third-degree polynomials, and more convoluted expressions involving additional parameters, among others. One of the goals of this paper is to demonstrate how to overcome these hindrances, particularly as the number of equations and/or dimensions increases. This is part of a larger effort to investigate pattern formation in the general Keller--Segel model, where analysis has mainly been conducted on systems of two equations. By increasing the model's dimension, one may obtain a more precise picture, but also introduce new complexity. However, with this complexity, new mathematical endeavors and insights emerge, and the authors hope that this paper will serve as a call and a guide to contribute to the field of pattern formation. 

{\bf Acknowledgements:}  This work was partially funded by NSF-DMS-2042413 and AFOSR MURI FA9550-22-1-0380.

\bibliographystyle{plain} 
\bibliography{ref}

\begin{thebibliography}{10}

\bibitem{amann2006hopf}
Herbert Amann.
\newblock Hopf bifurcation in quasilinear reaction-diffusion systems.
\newblock In Stavros Busenberg and Mario Martelli, editors, {\em Delay Differential Equations and Dynamical Systems}, pages 53--63, Berlin, Heidelberg, 1991. Springer Berlin Heidelberg.

\bibitem{berestycki2010self}
Henri Berestycki and Jean-Pierre Nadal.
\newblock Self-organised critical hot spots of criminal activity.
\newblock {\em European Journal of Applied Mathematics}, 21(4-5):371--399, 2010.

\bibitem{berestycki2014existence}
Henri Berestycki, Juncheng Wei, and Matthias Winter.
\newblock Existence of symmetric and asymmetric spikes for a crime hotspot model.
\newblock {\em SIAM Journal on Mathematical Analysis}, 46(1):691--719, 2014.

\bibitem{braga2001effects}
Anthony~A Braga.
\newblock The effects of hot spots policing on crime.
\newblock {\em Annals of the American Academy of Political and Social Science}, 578(1):104--125, 2001.

\bibitem{braga2014effects}
Anthony~A Braga, Andrew~V Papachristos, and David~M Hureau.
\newblock The effects of hot spots policing on crime: An updated systematic review and meta-analysis.
\newblock {\em Justice Quarterly}, 31(4):633--663, 2014.

\bibitem{buttenschoen2020cops}
Andreas Buttenschoen, Theodore Kolokolnikov, Michael~J Ward, and Juncheng Wei.
\newblock Cops-on-the-dots: The linear stability of crime hotspots for a 1-{D} reaction-diffusion model of urban crime.
\newblock {\em European Journal of Applied Mathematics}, 31(5):871--917, 2020.

\bibitem{cantrell2012global}
Robert~Stephen Cantrell, Chris Cosner, and Ra{\'u}l Man{\'a}sevich.
\newblock Global bifurcation of solutions for crime modeling equations.
\newblock {\em SIAM Journal on Mathematical Analysis}, 44(3):1340--1358, 2012.

\bibitem{chaturapruek2013crime}
Sorathan Chaturapruek, Jonah Breslau, Daniel Yazdi, Theodore Kolokolnikov, and Scott~G McCalla.
\newblock Crime modeling with {L\'e}vy flights.
\newblock {\em SIAM Journal on Applied Mathematics}, 73(4):1703--1720, 2013.

\bibitem{crandall1971bifurcation}
Michael~G Crandall and Paul~H Rabinowitz.
\newblock Bifurcation from simple eigenvalues.
\newblock {\em Journal of Functional Analysis}, 8(2):321--340, 1971.

\bibitem{crandall1973bifurcation}
Michael~G Crandall and Paul~H Rabinowitz.
\newblock Bifurcation, perturbation of simple eigenvalues and linearized stability.
\newblock {\em Archive for Rational Mechanics and Analysis}, 52:161--180, 1973.

\bibitem{d2015statistical}
Maria~R D'Orsogna and Matja{\v{z}} Perc.
\newblock Statistical physics of crime: A review.
\newblock {\em Physics of Life Reviews}, 12:1--21, 2015.

\bibitem{eck2003preventing}
John~E Eck.
\newblock Preventing crime at places.
\newblock In David~P. Farrington, Doris~Layton MacKenzie, Lawrence Sherman, and Brandon~C. Welsh, editors, {\em Evidence-based crime prevention}, pages 241--294. Routledge, 2003.

\bibitem{freitag2018global}
Marcel Freitag.
\newblock Global solutions to a higher-dimensional system related to crime modeling.
\newblock {\em Mathematical Methods in the Applied Sciences}, 41(16):6326--6335, 2018.

\bibitem{frydl2004fairness}
Kathleen Frydl and Wesley Skogan.
\newblock {\em Fairness and effectiveness in policing: The evidence}.
\newblock National Academies Press, 2004.

\bibitem{gai2024nucleation}
Chunyi Gai and Michael~J Ward.
\newblock The nucleation-annihilation dynamics of hotspot patterns for a reaction-diffusion system of urban crime with police deployment.
\newblock {\em SIAM Journal on Applied Dynamical Systems}, 23(3):2018--2060, 2024.

\bibitem{garcia2013existence}
M~Garcia-Huidobro, Raul Manasevich, and Jean Mawhin.
\newblock Existence of solutions for a 1-{D} boundary value problem coming from a model for burglary.
\newblock {\em Nonlinear Analysis: Real World Applications}, 14(5):1939--1946, 2013.

\bibitem{groff2019state}
Elizabeth~R Groff, Shane~D Johnson, and Amy Thornton.
\newblock State of the art in agent-based modeling of urban crime: An overview.
\newblock {\em Journal of Quantitative Criminology}, 35:155--193, 2019.

\bibitem{gu2017stationary}
Yu~Gu, Qi~Wang, and Guangzeng Yi.
\newblock Stationary patterns and their selection mechanism of urban crime models with heterogeneous near-repeat victimization effect.
\newblock {\em European Journal of Applied Mathematics}, 28(1):141--178, 2017.

\bibitem{heihoff2020generalized}
Frederic Heihoff.
\newblock Generalized solutions for a system of partial differential equations arising from urban crime modeling with a logistic source term.
\newblock {\em Zeitschrift f{\"u}r angewandte Mathematik und Physik}, 71(3):80, 2020.

\bibitem{jiang2024global}
Wenjing Jiang and Qi~Wang.
\newblock Global well-posedness and uniform boundedness of 2{D} urban crime models with nonlinear advection enhancement.
\newblock {\em European Journal of Applied Mathematics}, pages 1--13, 2024.

\bibitem{jones2010statistical}
Paul~A Jones, P~Jeffrey Brantingham, and Lincoln~R Chayes.
\newblock Statistical models of criminal behavior: The effects of law enforcement actions.
\newblock {\em Mathematical Models and Methods in Applied Sciences}, 20(supp01):1397--1423, 2010.

\bibitem{kleck2014more}
Gary Kleck and James~C Barnes.
\newblock Do more police lead to more crime deterrence?
\newblock {\em Crime \& Delinquency}, 60(5):716--738, 2014.

\bibitem{kolokolnikov2012stability}
Theodore Kolokolnikov, Michael~J Ward, and Juncheng Wei.
\newblock The stability of steady-state hot-spot patterns for a reaction-diffusion model of urban crime.
\newblock {\em Discrete and Continuous Dynamical Systems-B}, 19(5):1373--1410, 2014.

\bibitem{kumar2023modelling}
Manoj Kumar and Syed Abbas.
\newblock Modelling and prevention of crime using age-structure and law enforcement.
\newblock {\em Journal of Mathematical Analysis and Applications}, 519(2):126849, 2023.

\bibitem{ladyzhenskaya1968}
Olga~A Ladyzhenskaia and Nina~N Ural'tseva.
\newblock {\em Linear and Quasilinear Elliptic Equations}, volume~46 of {\em Mathematics in Science and Engineering}.
\newblock Academic Press, 1968.

\bibitem{levajkovic2016levy}
Tijana Levajkovic, Hermann Mena, and Martin Zarfl.
\newblock L{\'e}vy processes, subordinators and crime modelling.
\newblock {\em Novi Sad Journal of Mathematics}, 46(2):65--86, 2016.

\bibitem{lloyd2013localised}
David~JB Lloyd and Hayley O’Farrell.
\newblock On localized hotspots of an urban crime model.
\newblock {\em Physica D: Nonlinear Phenomena}, 253:23--39, 2013.

\bibitem{lloyd2016exploring}
David~JB Lloyd, Naratip Santitissadeekorn, and Martin~B Short.
\newblock Exploring data assimilation and forecasting issues for an urban crime model.
\newblock {\em European Journal of Applied Mathematics}, 27(3):451--478, 2016.

\bibitem{manasevich2013global}
Raul Manasevich, Quoc~Hung Phan, and Philippe Souplet.
\newblock Global existence of solutions for a chemotaxis-type system arising in crime modeling.
\newblock {\em European Journal of Applied Mathematics}, 24(2):273--296, 2013.

\bibitem{martinez2022bi}
Francisco~Javier Mart{\'\i}nez-Far{\'\i}as, Anah{\'\i} Alvarado-S{\'a}nchez, Eduardo Rangel-Cortes, and Arturo Hern{\'a}ndez-Hern{\'a}ndez.
\newblock Bi-dimensional crime model based on anomalous diffusion with law enforcement effect.
\newblock {\em Mathematical Modelling and Numerical Simulation with Applications}, 2(1):26--40, 2022.

\bibitem{mei2020existence}
Linfeng Mei and Juncheng Wei.
\newblock The existence and stability of spike solutions for a chemotaxis system modeling crime pattern formation.
\newblock {\em Mathematical Models and Methods in Applied Sciences}, 30(09):1727--1764, 2020.

\bibitem{mohler2012geographic}
George~O Mohler and Martin~B Short.
\newblock Geographic profiling from kinetic models of criminal behavior.
\newblock {\em SIAM Journal on Applied Mathematics}, 72(1):163--180, 2012.

\bibitem{mohler2011self}
George~O Mohler, Martin~B Short, P~Jeffrey Brantingham, Frederic~Paik Schoenberg, and George~E Tita.
\newblock Self-exciting point process modeling of crime.
\newblock {\em Journal of the American Statistical Association}, 106(493):100--108, 2011.

\bibitem{pan2018crime}
Chaohao Pan, Bo~Li, Chuntian Wang, Yuqi Zhang, Nathan Geldner, Li~Wang, and Andrea~L Bertozzi.
\newblock Crime modeling with truncated {L\'e}vy flights for residential burglary models.
\newblock {\em Mathematical Models and Methods in Applied Sciences}, 28(09):1857--1880, 2018.

\bibitem{pitcher2010adding}
Ashley~B Pitcher.
\newblock Adding police to a mathematical model of burglary.
\newblock {\em European Journal of Applied Mathematics}, 21(4-5):401--419, 2010.

\bibitem{ricketson2010continuum}
Lee Ricketson.
\newblock A continuum model of residential burglary incorporating law enforcement.
\newblock preprint on webpage at \url{www.academia.edu/570840/A_Continuum_Model_of_Residential_Burglary_Incorporating_Law_Enforcement}, 2010.

\bibitem{rodriguez2013global}
Nancy Rodr{\'\i}guez.
\newblock On the global well-posedness theory for a class of {PDE} models for criminal activity.
\newblock {\em Physica D: Nonlinear Phenomena}, 260:191--200, 2013.

\bibitem{rodriguez2021understanding}
Nancy Rodr{\'\i}guez, Qi~Wang, and Lu~Zhang.
\newblock Understanding the effects of on-and off-hotspot policing: {E}vidence of hotspot, oscillating, and chaotic activities.
\newblock {\em SIAM Journal On Applied Dynamical Systems}, 20(4):1882--1916, 2021.

\bibitem{rodriguez2022global}
Nancy Rodr{\'\i}guez and Michael Winkler.
\newblock On the global existence and qualitative behavior of one-dimensional solutions to a model for urban crime.
\newblock {\em European Journal of Applied Mathematics}, 33(5):919--959, 2022.

\bibitem{rosenbaum2006limits}
Dennis~P Rosenbaum.
\newblock The limits of hot spots policing.
\newblock In David Weisburd and Anthony~A. Braga, editors, {\em Police innovation: Contrasting perspectives}, pages 245--263. Cambridge University Press, 2006.

\bibitem{saldana2018age}
Joan Saldana, Maria Aguareles, Albert Aviny{\'o}, Marta Pellicer, and Jordi Ripoll.
\newblock An age-structured population approach for the mathematical modeling of urban burglaries.
\newblock {\em SIAM Journal on Applied Dynamical Systems}, 17(4):2733--2760, 2018.

\bibitem{sampson2004seeing}
Robert~J Sampson and Stephen~W Raudenbush.
\newblock Seeing disorder: Neighborhood stigma and the social construction of “broken windows”.
\newblock {\em Social Psychology Quarterly}, 67(4):319--342, 2004.

\bibitem{sherman1995general}
Lawrence~W Sherman and David Weisburd.
\newblock General deterrent effects of police patrol in crime “hot spots”: A randomized, controlled trial.
\newblock {\em Justice Quarterly}, 12(4):625--648, 1995.

\bibitem{shi2009global}
Junping Shi and Xuefeng Wang.
\newblock On global bifurcation for quasilinear elliptic systems on bounded domains.
\newblock {\em Journal of Differential Equations}, 246(7):2788--2812, 2009.

\bibitem{short2010nonlinear}
Martin~B Short, Andrea~L Bertozzi, and P~Jeffrey Brantingham.
\newblock Nonlinear patterns in urban crime: Hotspots, bifurcations, and suppression.
\newblock {\em SIAM Journal on Applied Dynamical Systems}, 9(2):462--483, 2010.

\bibitem{short2010dissipation}
Martin~B Short, P~Jeffrey Brantingham, Andrea~L Bertozzi, and George~E Tita.
\newblock Dissipation and displacement of hotspots in reaction-diffusion models of crime.
\newblock {\em Proceedings of the National Academy of Sciences}, 107(9):3961--3965, 2010.

\bibitem{short2008statistical}
Martin~B Short, Maria~R D'orsogna, Virginia~B Pasour, George~E Tita, Paul~J Brantingham, Andrea~L Bertozzi, and Lincoln~B Chayes.
\newblock A statistical model of criminal behavior.
\newblock {\em Mathematical Models and Methods in Applied Sciences}, 18(supp01):1249--1267, 2008.

\bibitem{skogan1999community}
Wesley~G Skogan and Susan~M Hartnett.
\newblock {\em Community policing, Chicago style}.
\newblock Oxford University Press, 1999.

\bibitem{thacher2011distribution}
David Thacher.
\newblock The distribution of police protection.
\newblock {\em Journal of Quantitative Criminology}, 27:275--298, 2011.

\bibitem{MATLAB}
{The MathWorks Inc.}
\newblock {MATLAB} version: 24.2.0.2740171 (r2024b) update 1, 2024.

\bibitem{PdeToolbox}
{The MathWorks Inc.}
\newblock Partial differential equation toolbox version: 24.2 (r2024b), 2024.

\bibitem{tse2016hotspot}
Wang~Hung Tse and Michael~J Ward.
\newblock Hotspot formation and dynamics for a continuum model of urban crime.
\newblock {\em European Journal of Applied Mathematics}, 27(3):583--624, 2016.

\bibitem{tse2018asynchronous}
Wang~Hung Tse and Michael~J Ward.
\newblock Asynchronous instabilities of crime hotspots for a 1-{D} reaction-diffusion model of urban crime with focused police patrol.
\newblock {\em SIAM Journal on Applied Dynamical Systems}, 17(3):2018--2075, 2018.

\bibitem{wang2020global}
Qi~Wang, Deqi Wang, and Yani Feng.
\newblock Global well-posedness and uniform boundedness of urban crime models: One-dimensional case.
\newblock {\em Journal of Differential Equations}, 269(7):6216--6235, 2020.

\bibitem{wang2017time}
Qi~Wang, Jingyue Yang, and Lu~Zhang.
\newblock Time-periodic and stable patterns of a two-competing-species {K}eller-{S}egel chemotaxis model: Effect of cellular growth.
\newblock {\em Discrete and Continuous Dynamical Systems-B}, 22(9):3547--3574, 2017.

\bibitem{weisburd2004can}
David Weisburd and John~E Eck.
\newblock What can police do to reduce crime, disorder, and fear?
\newblock {\em The annals of the American academy of political and social science}, 593(1):42--65, 2004.

\bibitem{weisburd1995policing}
David Weisburd and Lorraine Green.
\newblock Policing drug hot spots: {T}he {J}ersey {C}ity drug market analysis experiment.
\newblock {\em Justice Quarterly}, 12(4):711--735, 1995.

\bibitem{weisburd2010problem}
David Weisburd, Cody~W Telep, Joshua~C Hinkle, and John~E Eck.
\newblock Is problem-oriented policing effective in reducing crime and disorder? {F}indings from a {C}ampbell systematic review.
\newblock {\em Criminology \& Public Policy}, 9(1):139--172, 2010.

\bibitem{windows1982police}
James~Q. Wilson and George~L. Kelling.
\newblock Broken windows: {T}he police and neighborhood safety.
\newblock {\em Atlantic Monthly}, 249(3):29--38, 1982.

\bibitem{winkler2019global}
Michael Winkler.
\newblock Global solvability and stabilization in a two-dimensional cross-diffusion system modeling urban crime propagation.
\newblock {\em Annales de l'Institut Henri Poincaré C, Analyse non linéaire}, 36(6):1747--1790, 2019.

\bibitem{yerlanov2025}
Madi Yerlanov, Qi~Wang, and Nancy Rodr{\'\i}guez.
\newblock Formation and suppression of hotspots in urban crime models with law enforcement.
\newblock {\em Chaos: An Interdisciplinary Journal of Nonlinear Science}, 35(6), 2025.

\bibitem{zipkin2014cops}
Joseph~R Zipkin, Martin~B Short, and Andrea~L Bertozzi.
\newblock Cops on the dots in a mathematical model of urban crime and police response.
\newblock {\em Discrete and Continuous Dynamical Systems-B}, 19(5):1479--1506, 2014.

\end{thebibliography}

\end{document}